%% file: haldanelong_submitted.tex
\documentclass[11pt]{article}
\usepackage{amsmath,amsthm,amsfonts,amssymb, graphicx,mathabx,epsfig}
\usepackage{bbm}
\usepackage{bm}
\usepackage{authblk}
\usepackage{setspace}
\usepackage[usenames]{color}
\usepackage[active]{srcltx}

\usepackage{graphicx,color}
\usepackage{footmisc}
\usepackage{amsmath}
\usepackage{amsfonts}
\usepackage{amssymb}
\usepackage{bbm}
\usepackage{epstopdf}
\usepackage{color}
\usepackage{dsfont}
\numberwithin{equation}{section}
\usepackage{tikz}
\usepackage{pgfplots}
\usepackage{subfig}
\usetikzlibrary{fit}
\usetikzlibrary{shapes.geometric}
\usetikzlibrary{decorations.pathmorphing}
\usepackage{slashed}
\usepackage{color}

\renewcommand{\d}{\mathrm{d}}

\newcommand{\e}{\varepsilon}
\newcommand{\dbar}{\kern-.1em{\raise.8ex\hbox{ -}}\kern-.6em{d}}

\def\?{\marginpar{not sure}}

\newcommand{\comment}[1]{}

\newtheorem{thm}{Theorem}[section]

\newtheorem{lemma}[thm]{Lemma}
\newtheorem{prop}[thm]{Proposition}

\def \be{\begin{equation}}
\def \ee{\end{equation}}
\def \ben{\begin{equation*}}
\def \een{\end{equation*}}
\def \bea{\begin{eqnarray}}
\def \eea{\end{eqnarray}}

\def\qed{\hfill\raise1pt\hbox{\vrule height5pt width5pt depth0pt}}
\def\nn{\nonumber}

\def\Tr{\mathrm{Tr}}

\def\L{\Lambda}
\def\l{\lambda}
\def\r{\rho}
\def\s{\sigma}
\def\a{\alpha}
\def\b{\beta}
\def\d{\delta}
\def\D{\Delta}
\def\m{\mu}
\def\g{\gamma}

\def\e{\varepsilon}
\def\t{\tau}
\def\th{\theta}

\definecolor{light}{gray}{.75}

\let\a=\alpha \let\b=\beta  \let\g=\gamma  \let\d=\delta \let\e=\varepsilon
     \let\th=\theta  \let\l=\lambda
\let\m=\mu    \let\n=\nu    \let\x=\xi     \let\p=\pi    \let\r=\rho
\let\s=\sigma \let\t=\tau    \let\c=\chi
   \let\o=\omega
 \let\D=\Delta  \let\L=\Lambda 
         
\let\O=\Omega

\newcommand{\VV}{{\mathcal V}}

\newcommand{\qq}{{\bf q}}
\newcommand{\pp}{{\bf p}}

\newcommand{\xx}{{\bf x}}
\newcommand{\yy}{{\bf y}}

\newcommand{\kk}{{\bf k}}

\newcommand{\RR}{{\mathcal R}}

\newcommand{\WW}{{\mathcal W}}
\newcommand{\LL}{{\mathcal L}}

\def\nn{\nonumber}

\def\\{\hfill\break}
\def\={:=}

\let\0=\noindent

\def\sign{{\rm sign}}

\def\tende#1{\,\vtop{\ialign{##\crcr\rightarrowfill\crcr\noalign{\kern-1pt
    \nointerlineskip} \hskip3.pt${\scriptstyle #1}$\hskip3.pt\crcr}}\,}
\def\otto{\,{\kern-1.truept\leftarrow\kern-5.truept\to\kern-1.truept}\,}

\def\to{\rightarrow}

\def\qed{\hfill\raise1pt\hbox{\vrule height5pt width5pt depth0pt}}

\def\V#1{{\bf#1}}
\def\be{\begin{equation}}
\def\ee{\end{equation}}
\def\bea{\begin{eqnarray}}
\def\eea{\end{eqnarray}}
\def\nn{\nonumber}

\def\Tr{\mathrm{Tr}}

\theoremstyle{plain}

\theoremstyle{definition}

\begin{document}

\title{Quantization of the interacting Hall conductivity\\ in the critical regime}

\author[1]{Alessandro Giuliani}
\affil[1]{University of Roma Tre, Department of Mathematics and Physics, L.go S. L. Murialdo 1, 00146 Roma, Italy}
\author[2]{Vieri Mastropietro}
\affil[2]{University of Milano, Department of Mathematics ``F. Enriquez'', Via C. Saldini 50, 20133 Milano, Italy}
\author[3]{Marcello Porta}
\affil[3]{University of T\"ubingen, Department of Mathematics, Auf der Morgenstelle 10, 72076 T\"ubingen, Germany}

\maketitle

\begin{abstract} 
The Haldane model is a paradigmatic $2d$ lattice model exhibiting the integer quantum Hall effect. We consider an interacting version of the model, and prove that for short-range interactions, smaller than the bandwidth, 
the Hall conductivity is quantized, for all the values of the parameters outside two critical curves, across which the model undergoes a `topological' phase transition: the Hall coefficient 
remains integer and constant as long as we continuously deform the parameters without crossing the curves; when this happens, the Hall coefficient jumps abruptly to a different 
integer. Previous works were limited to the perturbative regime, in which the interaction is much smaller than the bare gap, so they were restricted 
to regions far from the critical lines. The non-renormalization of the Hall conductivity arises as a 
consequence of lattice conservation laws and of the regularity properties of the current-current correlations. Our method provides a full construction of the critical curves, which 
are modified (`dressed') by the electron-electron interaction. 
The shift of the transition curves manifests itself via apparent infrared divergences in the naive perturbative series, which we resolve
via renormalization group methods.
\end{abstract}
\maketitle

\section{Introduction}

One of the remarkable features of the Integer Quantum Hall Effect (QHE) is the impressive precision of the quantization of the plateaus observed 
in the experiments. While the experimental samples have a very complex microscopic structure, depending on a huge number 
of non-universal details related to molecular forces and the atomic structure, the conductance
appears to be quantized at a very high precision, and the result only depends on fundamental constants.  
The understanding of this phenomenon, via a connection between the Hall conductivity and a 
topological invariant \cite{ASS1,TKNN} was a major success of theoretical condensed matter in the 80s. The argument was later generalized to non-interacting disordered systems
\cite{AG,ASS2,B,BES} and to clean multi-particle systems \cite{AS,N}: however, the definition of conductivity in the interacting case required the presence of an unphysical averaging over fluxes, expected to be unimportant 
in the thermodynamic limit, but a proof remained elusive for many years. 
Arguments based on Ward Identities for Quantum ElectroDynamics in $(2+1)$-dimensions
\cite{CH, H1talk, IM}, or on the properties of anomalies \cite{F}, offered an alternative view on the QHE: they indicated that quantization should persists 
in the presence of many body interaction, but such conclusions were based on manipulations of divergent series, or of effective actions arising in a formal scaling limit. 

The problem of a mathematical proof of the quantization of the Hall conductivity in the presence of many-body interactions remained open for several years. After the works
\cite{AG,AS,ASS2,B,BES,N}, it was dormant for more than a decade, and then, in recent years, it was actively reconsidered again.
The Hall conductivity of a class of interacting fermionic systems was finally computed, and shown to be quantized, in \cite{GMPhall}. 
The fermionic Hamiltonians considered in \cite{GMPhall} have the form $H_0+\l V$,
where: $H_0$ is quadratic and gapped, with the chemical potential in the middle of a spectral gap of width $\D_0$; $V$ is a many body interaction, and $\l$ is its strength. 
The Hall coefficient is written as a power series in $\l$, which is well known to be convergent \cite{C, GK, Br, BK, L, AR}, 
provided that $|\l|\ll\D_0$. Ward Identities
ensure the cancellation of the possible interaction corrections to the conductivity, thus proving that the Hall coefficient is equal to the reference non-interacting value, which is well known to be quantized.
A similar result was also obtained by different methods in \cite{HM,H}, by combining a theorem guaranteeing the quantization of the Hall conductance under the {\it assumption} of a volume-independent spectral gap, see \cite{HM}, with the proof of persistence of the spectral gap in fermionic systems of the form
$H_0+\l V$, with $|\l|\ll\D_0$, see \cite{H}. See also \cite{BBdRF} and \cite{dRS}.

The above results leave the problem of computing the Hall conductivity in the opposite, and more interesting, regime of interactions larger than the reference non-interacting gap
completely open: if $|\l|\gg\D_0$, approaches based on naive perturbation theory break down.
Based on the ideas of quasi-adiabatic evolution of the ground state \cite{HM}, one knows that if the interaction does not close the spectral gap above the ground state, then the Hall coefficient remains the same as the one of $H_0$, even in this case. However, the verification
that the gap persist is much harder in this non-perturbarive regime.

For definiteness, we investigate this problem in the case that $H_0$ is the Haldane model \cite{H1}, which is a paradigmatic $2d$ lattice model exhibiting the integer quantum Hall effect, but our approach applies to more general situations. 
The model is characterized by distinct gapped topological phases, separated by critical lines along which the gap closes. We consider an interacting version of the Haldane model, described by a Hamiltonian $H_0+\l V$; previous results, see  \cite{GMPhall} and   \cite{HM,H}, guarantee that the Hall coefficient is independent of the interaction only far from the critical lines, well inside the insulating regions.
In this paper, we compute the Hall conductivity for all the values of the parameters outside two critical curves, across which the model undergoes a `topological' phase transition: the Hall coefficient 
remains integer and constant as long as we continuously deform the parameters without crossing the curves; when this happens, the Hall coefficient jumps abruptly to a different 
integer. The main difficulties are related to the fact that the critical lines are non-universal and interaction-dependent, thus making a naive perurbative approach unreliable. This is
analogous to what happens in the theory of critical phenomena, where the critical temperature is modified by the interaction, and one needs to tune the temperature as the interaction is 
switched on, in order to stay at criticality. Technically, we proceed in a similar way; we do not expand around the non interacting Hamiltonian but around a reference quadratic Hamiltonian, characterized by the same gap as the interacting system, whose value is fixed self-consistently. It is an interesting question to see whether a similar strategy can be applied within the Hamiltonian scheme of \cite{H}.

Our results extend and complement those of \cite{GJMP}, where we considered the same model (in the special case of ultra-local interactions) and we proved: (i) 
existence of the critical curves, but without an explicit control on their regularity properties, and (ii) universality of the jump in the Hall coefficient across the critical curves, but without a 
proof that the coefficient remains constant in each connected component of the complement of the  critical curves. Combining the results of \cite{GJMP} with those presented here, 
we have a complete construction of the topological phase diagram of the interacting Haldane model. We stress that the choice of this specific model is made only for the sake of 
definiteness. The proof applies to any short-range, translation invariant lattice fermionic Hamiltonian, whose bare energy bands display conical intersections at criticality.

Our presentation is organized as follows: in Section \ref{sec:main} we define the class of interacting Haldane models that we are going to consider, and we state our main result. In Section \ref{sec:proof} we prove the quantization of 
the Hall coefficient, under suitable regularity assumptions on the Euclidean correlation functions of the interacting model. 
This part of the paper holds in great generality, for a class of interacting fermionic systems much larger than the interacting Haldane model. 
In Section \ref{sec.4a} we prove the regularity assumptions on the correlations for the model at hand, via rigorous renormalization group methods. In Section \ref{sec.5} we put things together and complete the proof of our main 
result. 

\section{Main result}\label{sec:main}
\subsection{The model}\label{sec.2.1}
The Haldane model describes fermions on the honeycomb lattice hopping on nearest and next-to-nearest neighbours, in the presence of a transverse magnetic field, with {\it zero net flux} through the hexagonal cell, and of a staggered potential. 

Let $\L=\big\{\vec x \mid \vec x = n_{1} \vec \ell_{1} + n_{2} \vec \ell_{2},\; n_{i} \in \mathbb{Z}\}\subset\mathbb{R}^{2}$ be the infinite triangular lattice generated by the 
two basis vectors $\vec \ell_{1} = \frac{1}{2}(3,\, -\sqrt{3}), \vec \ell_{2} = \frac{1}{2}(3, \sqrt{3})$. Given $L\in \mathbb{N}$, we also let $\L_{L}=\L/ L\L$ be the corresponding finite periodic triangular lattice of side $L$, which will be identified with the set
$\L_L= \big\{ \vec x \mid \vec x = n_{1} \vec \ell_{1} + n_{2} \vec \ell_{2},\; n_{i} \in \mathbb{Z}\cap[0,L) \big\}$ 
with periodic boundary conditions. The lattice is endowed with the Euclidean distance on the torus, denoted by $| \vec x - \vec y|_L=\min_{m\in\mathbb Z^2}| \vec x - \vec y+m_1 \vec \ell_1 L+m_2 \vec \ell_2 L|$.
The number of sites of $\Lambda_{L}$ is $|\L_{L}| = L^{2}$. The periodic honeycomb lattice can be realized as the superposition of two periodic triangular sublattices $\L^{\text{A}}_{L} \equiv \L_{L}$, $\L^{\text{B}}_{L}\equiv \L_{L} + (1,0)$. Equivalently, we can think the honeycomb lattice as a triangular lattice, with two internal degrees of freedom corresponding to the $A, B$ sublattices.

It is convenient to define the model in second quantization. The one-particle Hilbert space is the set of functions $\mathfrak{h}_{L} = \{ f: \L_{L}\times \{ \uparrow, \downarrow \}\times \{A, B\}\to \mathbb{C}  \} \simeq \mathbb{C}^{L^{2}}\otimes \mathbb{C}^{4}$. We define the fermionic Fock space $\mathcal{F}_{L}$ as:
\be
\mathcal{F}_{L} = \mathbb{C} \oplus \bigoplus_{n = 1}^{4L^{2}} \mathcal{F}_{L}^{(n)}\;,\qquad \mathcal{F}_{L}^{(n)} = \mathfrak{h}_{L}^{\wedge n}\;,
\ee
with $\wedge$ the antisymmetric tensor product. The number $4$ stands for the number of sublattices times the number of spin degrees of freedom. Notice that for fixed $L$, $\mathcal{F}_{L}$ is a finite-dimensional space. For a given site $\vec x\in \L_{L}$, we introduce fermionic creation and annihilation operators $\psi^{\pm}_{\vec x, \r, \s}$, with $\r \in \{A, B\}$ the sublattice label and $\s \in \{\uparrow, \downarrow \}$ the spin label. They satisfy the standard canonical anticommutation relations $\{ \psi^{+}_{\vec x, \r, \s}, \psi^{-}_{\vec y, \r', \s'}\} = \d_{\r,\r'} \d_{\s,\s'} \d_{\vec x,\vec y}$ and $\{ \psi^{+}_{\vec x, \r, \s}, \psi^{+}_{\vec y, \r', \s'}\} = \{ \psi^{-}_{\vec x, \r, \s}, \psi^{-}_{\vec y, \r', \s'}\} = 0$. The operators $\psi^{\pm}_{\vec x,\r,\s}$ are consistent with the periodic boundary conditions on $\L_{L}$, $\psi^{\pm}_{\vec x + n_{1} L + n_{2} L,\r,\s} = \psi^{\pm}_{\vec x,\r,\s}$.

The reciprocal lattice $\L_{L}^{*}$ of $\L_{L}$ is the triangular lattice generated by the basis vectors $\vec G_{1}$, $\vec G_{2}$, such that $\vec G_i\cdot\vec \ell_j=2\pi \d_{i,j}$. Explicitely, 
$\vec G_{1} = \frac{2\pi}{3}(1,\, -\sqrt{3})$, $\vec G_{2} = \frac{2\pi}{3}(1,\, \sqrt{3})$. We define the finite-volume Brillouin zone as $\mathcal{B}_{L} := \Big\{ \vec k \in \mathbb{R}^{2} \mid \vec k = \frac{n_1}L \vec G_{1} + \frac{n_2}L \vec G_{2},\; n_{i}\in \mathbb{Z}\cap[0,L) \Big\}$. We define the Fourier transforms of the fermionic creation and annihilation operators as:
\be
\psi^{\pm}_{\vec x,\r,\s} = \frac{1}{L^{2}} \sum_{\vec k \in \mathcal{B}_{L}} e^{\pm i \vec k\cdot \vec x} \hat \psi_{\vec k, \r,\s} \quad \forall \vec x\in \L_{L} \qquad \Longleftrightarrow \qquad  \hat \psi^{\pm}_{\vec k, \r, \s} = \sum_{\vec x\in \L_{L}} e^{\mp i\vec k\cdot \vec x} \psi_{\vec x, \r, \s}\quad \forall \vec k\in \mathcal{B}_{L}\;. 
\ee
With this definition, $\hat \psi^\pm_{\vec k,\r,\s}$ are periodic over the Brillouin zone, $\hat \psi_{\vec k + m_{1} \vec G_{1} + m_{2} \vec G_{2},\r,\s} = \hat \psi_{\vec k,\r,\s}$. Moreover, the Fourier transforms of the fermionic operators satisfy the anticommutation relations: $\{ \hat\psi^{+}_{\vec k, \r, \s}, \hat\psi^{-}_{\vec k', \r', \s'} \} = L^{2} \d_{\vec k, \vec k'}\d_{\r,\r'}\d_{\s,\s'}$ and $\{ \hat\psi^{+}_{\vec k, \r, \s}, \hat\psi^{+}_{\vec k', \r', \s'} \} = \{ \hat\psi^{-}_{\vec k, \r, \s}, \hat\psi^{-}_{\vec k', \r', \s'} \} = 0$.

The Hamiltonian of the model is: $\mathcal{H} = \mathcal{H}_0 + U\mathcal{V}$, with $\mathcal{H}_0$ the noninteracting Hamiltonian and $\mathcal{V}$ the many-body interaction
of strength $U$. We have:
\bea\label{eq:Haldane}
\mathcal{H}_0 &=& -t_{1} \sum_{\vec x \in \L_{L}}\sum_{\s = \uparrow, \downarrow} [ \psi^{+}_{\vec x, A, \s} \psi^{-}_{\vec x, B, \s} + \psi^{+}_{\vec x, A, \s} \psi^{-}_{\vec x -\vec \ell_{1}, B, \s} + \psi^{+}_{\vec x, A, \s} \psi^{-}_{\vec x - \vec \ell_{2}, B, \s} + \text{h.c.} ] \nn\\
&& - t_{2} \sum_{\vec x \in \L_{L}}\sum_{\substack{\alpha = \pm \\ j=1,2,3}}  \sum_{\s=\uparrow\downarrow} [ e^{i\alpha \phi} \psi^{+}_{\vec x,A,\s}\psi^{-}_{\vec x + \a\vec \g_{j}, A, \s} + e^{-i\alpha\phi}\psi^{+}_{\vec x,B,\s}\psi^{-}_{\vec x + \a\vec \g_{j}, B, \s}  ]\nn\\
&& + W \sum_{\vec x\in \L_{L}}\sum_{\s = \uparrow\downarrow} [n_{\vec x, A, \s} - n_{\vec x, B, \s}] - \mu \sum_{\vec x\in \L_{L}} \sum_{\r = A, B} \sum_{\s = \uparrow\downarrow} n_{\vec x, \r, \s}\;,
\eea
with $\vec \g_{1} = \vec \ell_{1} - \vec \ell_{2}$, $\vec \g_{2} = \vec \ell_{2}$, $\vec \g_{3} = -\vec \ell_{1}$ and $n_{\vec x,\r,\s} = \psi^{+}_{\vec x,\r,\s}\psi^{-}_{\vec x,\r,\s}$. For definiteness, we assume that $t_{1} > 0$ and $t_{2} > 0$.  The 
term proportional to $t_1$ describes nearest neighbor hopping on the hexagonal lattice. The term 
proportional to $t_2$ describes next-to-nearest neighbor hopping,
with the complex phases $e^{\pm i\phi}$ modeling the effect of an external, transverse, magnetic field. The term 
proportional to $W$ describes a staggered potential, favoring the occupancy of the $A$ or $B$ 
sublattice, depending on whether $W$ is negative or positive. Finally, the term proportional to $\mu$ is the chemical potential, which controls the average particle
density in the Gibbs state. See Fig. \ref{fig:haldane}.
\begin{figure}[hbtp]
\centering
\def\svgwidth{160pt}
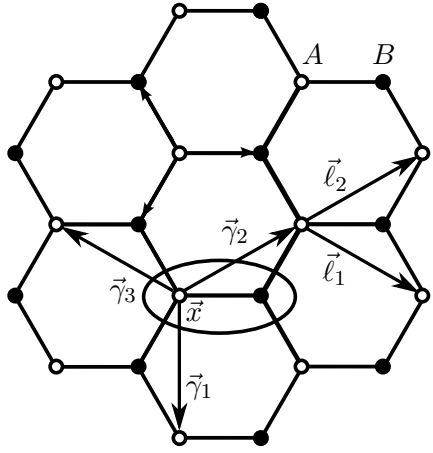
\caption{The honeycomb lattice of the Haldane model. The empty dots belong to $\L^{\text{A}}_{L}$, while the black dots belong to $\L^{\text{B}}_{L}$. The oval encircles the two sites of the fundamental cell, labeled by the position of the empty dot, i.e., of the site of the $A$ sublattice. The nearest neighbor 
vectors $\vec \d_i$, are shown explicitly, together with the next-to-nearest neighbor vectors $\vec \g_i$, and the two basis vectors $\vec \ell_{1,2}$ of $\L_L$.
}\label{fig:haldane}
\end{figure}
 Concerning the many-body interaction, we assume it to be a density-density interaction of the form:
\be\label{eq:V}
\mathcal{V} = \sum_{\vec x,\vec y\in \L_{L}}\sum_{\r = A, B} (n_{\vec x, \r}-1) v_{\r\r'}(\vec x-\vec y) (n_{\vec y,\r'}-1)\;,
\ee
where $n_{\vec x, \r}=\sum_{\s=\uparrow,\downarrow}n_{\vec x,\r,\s}$ and $v_{\r\r'}(\vec x)=v\big(\vec x+(\d_{\r,B}-\d_{\r',B})(1,0)\big)$, 
with $v$ a finite range, rotationally invariant, potential. 

The noninteracting Hamiltonian can be rewritten as:
\be
\mathcal{H}_0= \sum_{\vec x,\vec y} \sum_{\r, \r', \s} \psi^{+}_{\vec x, \r, \s} H_{\r\r'}(\vec x- \vec y) \psi^{-}_{\vec y, \r', \s}\;,
\ee
where $H_{\r\r'}(\vec x- \vec y)$ are the matrix elements of the Haldane model; we denote by $H(\vec x- \vec y)$ the corresponding $2\times2$ block. 
We introduce the Bloch Hamiltonian $\hat H(k) = \sum_{\vec z\in \L_{L}} e^{-i\vec k\cdot \vec z} H(\vec z)$, with $\vec k \in \mathcal{B}_{L}$. An explicit computation gives:
\be
\hat H( \vec k) = \begin{pmatrix} -2t_2\a_1(\vec k)\cos\phi+ m(\vec k)-\m    & -t_{1} \Omega^*( \vec k) \\ - t_{1}\Omega( \vec k) & -2t_2\a_1(\vec k)\cos\phi- m(\vec k)-\m   \end{pmatrix}\ee
where:
\begin{equation}\label{eq:Haldane2}
\begin{split}
&\alpha_{1}(\vec k) = \sum_{j=1}^3\cos(\vec k\cdot\vec \gamma_j)\;,\qquad m(\vec k) = W - 2t_{2}\sin\phi\, \alpha_{2}(\vec k)\;,
\\
&\alpha_{2}(\vec k) = \sum_{j=1}^3\sin(\vec k\cdot\vec \gamma_j)
\;,\qquad \Omega(\vec k) = 1 + e^{-i\vec k\cdot\vec \ell_1} + e^{-i\vec k\cdot\vec\ell_2}\;.
\end{split}
\end{equation}

The corresponding energy bands are 
\begin{equation}
\e_{\pm}(\vec k) =-2t_2\a_1(\vec k)\cos\phi \pm \sqrt{m(\vec k)^{2} + t_1^{2}|\Omega(\vec k)|^{2}}\;.\nn
\end{equation}
To make sure that the energy bands do not overlap, we assume that $t_2/t_1<1/3$. For $L\to \infty$,
the two bands can touch only at the {\it Fermi points} $\vec k_{F}^{\pm} = \big( \frac{2\pi}{3}, \pm\frac{2\pi}{3\sqrt3} \big)$, which are the two zeros of $\Omega(\vec k)$,
around which $\O(\vec k_F^\pm+\vec k')\simeq \frac32(ik_1'\pm k_2')$.
The condition that the two bands touch
at $\vec k_F^\o$, with $\o = +,-$, is that $m_\o=0$, with 
\be
m_{\o} \equiv m(\vec k_{F}^{\o}) = W +\o 3\sqrt{3}\,t_{2}\sin\phi \;.\label{eq.mpm}
\ee
If, instead, $m_+$ and $m_-$ are both different from zero, then the spectrum of $\hat H(\vec k)$ is gapped for all $\vec k$, corresponding to an insulating phase.

\subsection{Lattice currents and linear reponse theory}

Let $n_{\vec x} = \sum_{\r = A,B} \sum_{\s=\uparrow,\downarrow}n_{\vec x, \r, \s}$ be the total density operator at $\vec x$. Its time-evolution is given by $n_{\vec x}(t) = e^{i\mathcal{H} t} n_{\vec x} e^{-i\mathcal{H} t}$, which satisfies the following {\it lattice continuity equation}:
\be\label{eq:cont}
\partial_{t} n_{\vec x}(t) = i[ \mathcal{H}, n_{\vec x}(t) ] \equiv \sum_{\vec y}j_{\vec x,\vec y}(t)\;,
\ee
with $j_{\vec x,\vec y}$ the {\it bond current}:
\be
j_{\vec x,\vec y} =  \sum_{\r,\r'=A,B}\ \sum_{\s=\uparrow,\downarrow} (i\psi^{+}_{\vec y, \r',\s} H_{\r'\r}(\vec y-\vec x) \psi^{-}_{\vec x,\r,\s} + \text{h.c.})\;. 
\ee
Notice that $j_{\vec x,\vec y} = -j_{\vec y,\vec x}$. Thus, using that $H(\vec x) \neq 0$ if and only if $\vec x = \vec 0, \pm \vec \ell_{1},\pm\vec\ell_2, \pm (\vec \ell_{1} - \vec \ell_{2})$, Eq. (\ref{eq:cont}) implies:
\bea
\partial_{t} n_{\vec x}(t) &=& \sum_{\vec y} j_{\vec x,\vec y}(t) = \sum_{i=1,2} [j_{\vec x, \vec x+ \vec \ell_{i}} + j_{\vec x, \vec x - \vec \ell_{i}}] + j_{\vec x, \vec x+ \vec \ell_{1} - \vec \ell_{2}} + j_{\vec x, \vec x - \vec \ell_{1} + \vec \ell_{2}}\nn\\ &\equiv& -\text{d}_{1} \tilde{\text{\j}}_{1, \vec x} - \text{d}_{2} \tilde{\text{\j}}_{2, \vec x}\;,\label{eq.2.11}
\eea
where $\text{d}_{i} f(\vec x) = f(\vec x) - f(\vec x - \vec \ell_{i})$ is the lattice derivative along the $\vec \ell_{i}$ direction, and:
\be\label{eq.jj}
\tilde{\text{\j}}_{1,\vec x} = -j_{\vec x, \vec x+ \vec \ell_{1}} - j_{\vec x, \vec x+ \vec \ell_{1} - \vec \ell_{2}}\;,\qquad \tilde{\text{\j}}_{2,\vec x} = -j_{\vec x, \vec x + \vec \ell_{2}} - j_{\vec x, \vec x - \vec \ell_{1} + \vec \ell_{2}}\;.
\ee
The operators $\tilde{\text{\j}}_{i, \vec x}$ are the components along the $\vec \ell_{i}$ directions of the total vectorial current, defined as
\be \vec \jmath_{\vec x} =  \tilde{\text{\j}}_{1,\vec x} \vec \ell_{1}+  \tilde{\text{\j}}_{2,\vec x}\vec \ell_{2}\;.\label{totcur}\ee
Note that, given the definitions of $\vec\ell_{1,2}$, the components of the lattice current along the two reference, orthogonal, coordinate directions are: 
\be \label{j12} j_{1,\vec x}=\frac32(\tilde{\text{\j}}_{1,\vec x}+\tilde{\text{\j}}_{2,\vec x}),\qquad j_{2,\vec x}=\frac{\sqrt3}2(-\tilde{\text{\j}}_{1,\vec x}+\tilde{\text{\j}}_{2,\vec x}).\ee

We are interested in the transport properties of the Haldane-Hubbard model, in the linear response regime. The {\it Gibbs state} of the interacting model is defined as: $\langle \cdot \rangle_{\b, L} = \Tr_{\mathcal{F}_{L}} \cdot e^{-\b \mathcal{H}} / \mathcal{Z}_{\b, L}$ with $\mathcal{Z}_{\b, L} = \Tr_{\mathcal{F}_{L}} e^{-\b \mathcal{H}}$ the partition function. We define the conductivity matrix via the {\it Kubo formula}, for $i, j =1,2$:
\be\label{eq:realkubo}
\s_{ij} := \frac1{|\vec\ell_1\wedge\vec\ell_2|}\lim_{p_0 \to 0^{+}} \frac{1}{p_0}\Big[-i \int_{-\infty}^{0} dt\, e^{p_0 t} \lim_{\beta, L\to \infty} \frac{1}{L^{2}} \langle [ \mathcal{J}_{i}\,, \mathcal{J}_{j}(t) ] \rangle_{\b, L} + i
 \lim_{\beta, L\to \infty} \frac{1}{L^{2}} \langle [\mathcal J_i,\mathcal X_j] \rangle_{\b, L}\Big]\;,
\ee
with $\vec{\mathcal{X}}$ the second quantization of the position operator, and $\vec{\mathcal{J}} = \sum_{\vec x\in\L_L} \vec \jmath_{\vec x}=i[\mathcal H,\vec{\mathcal{X}}]$ 
the total current operator\footnote{\label{ciaciao}There is an issue in defining the position operator on the torus. In order to avoid the problem, we interpret the second term in \eqref{eq:realkubo} as 
being equal to $i\frac{\partial }{\partial q_j}\pmb{\langle} [ \mathcal{J}_{i}(\vec q), N(-\vec q)]\pmb{\rangle}_\infty\big|_{\vec q=\vec 0}$, where: $\vec{\mathcal{J}}(\vec q)=
\sum_{\vec x\in\L_L} \vec j_{\vec x}e^{i\vec q\vec x}$,
$N(\vec q)=\sum_{\vec x\in\L_L} n_{\vec x}e^{i\vec q\vec x}$, $\pmb{\langle}  [ \mathcal{J}_{i}(\vec q), N(-\vec q)]\pmb{\rangle}_\infty:=
\lim_{\b,L\to\infty}\frac1{L^2}\langle [ \mathcal{J}_{i}(\vec q^{(L)}), N(-\vec q^{(L)})]\rangle_{\b,L}$, and $\vec q^{(L)}$ a sequence of vectors in $\mathcal B_L$ such that 
$\lim_{L\to\infty}\vec q^{(L)}=\vec q$.}. This formula describes the linear response of the average current at the time $t= 0$ an adiabatic external field of the form $e^{\eta t} \vec E\cdot \vec{\mathcal{X}}$, see {\it e.g.} \cite{Ale} for a formal derivation, and \cite{BdRF1,BdRF2,MT,T} for a rigorous derivation in a slightly different setting. 


\medskip

{\bf Remark.} The indices $i,j$ labelling the elements of the conductivity matrix \eqref{eq:realkubo} refer to the two reference,  
orthogonal, coordinate directions. Sometimes, a similar definition of the Kubo matrix is given, where, instead, the indices $i,j$ label the two lattice coordinate directions 
$\vec\ell_1,\vec\ell_2$ (`adapted basis'). The two definitions are, of course, related in a simple way, via the transformation induced by the change of basis. In particular, the 
transverse conductivities defined in the orthogonal and in the adapted basis are the same, up to an overall multiplicative factor, equal to $|\vec \ell_1\wedge\vec\ell_2|$. 
The longitudinal conductivities are, instead, related via a matrix relation that mixes up the diagonal and non-diagonal components of the conductivity matrix. For ease of 
comparison with experimental papers on graphene, or graphene-like materials, we prefer to use the 
definition involving the orthogonal reference directions, which we find more natural. 

\medskip

In the absence of interactions, the Kubo conductivity matrix of the Haldane model can be computed explicitly. Suppose that $m_{\o} \neq 0$, both for $\o=+$ and for $\o=-$, 
and let us choose the chemical potential in the spectral gap. For instance, let $\m = -2t_{2} \cos\phi \alpha_{1}(k_{F}^{\o})$, which corresponds to choosing the chemical potential in the `middle of the gap'. Then, it turns out that \cite{H1}:
\be
\s_{11} = 0\;,\qquad \s_{12} = -\s_{21} = \frac{\nu}{2\pi}\;,\qquad \nu = \sign(m_{+}) - \sign(m_{-})\;.
\ee
The integer $\nu$ is the Chern number of the Bloch bundle associated to $\hat H(\vec k)$. The zeros of $m_{\o}=W+\o3\sqrt3 t_2\sin\phi$, with $\o\in\{\pm\}$, 
define the {\it critical curves} of the Haldane model, which separate the different topological phases, corresponding to different values of $\n$. {\it On} the curves, the spectrum is 
gapless: the energy bands intersect with conical intersection, and the system displays a quantization phenomenon of the {\it longitudinal} conductivity:
\be\label{eq:univ11}
\s_{11} = \s_{22} = \frac{1}{8}\;,
\ee
while $\s_{11} =\s_{22}= \frac{1}{4}$ at the `graphene points' $m_{+} = m_{-} = 0$. 

\subsection{Main result: interacting topological phases and phase transitions}\label{sec:mr}

Let us now turn on the many-body interaction, $U\neq 0$. In previous works, it was proved that the quantization of the conductivity persists, but 
only for interactions of strength {\it much smaller than the gap of $\mathcal H_0$}.
Our main result, summarized in the next theorem, overcomes this limitation. 
\begin{thm}\label{thm:1} There exists $U_{0} >0$, independent of $W,\phi$, such that for $|U| < U_0$ the following is true. There exist three functions,
$\d_{\o}(U,W, \phi)$, with $\o=\pm$, and $\xi(U,W, \phi)$, analytic in $U$ and continuously differentiable in $W,\phi$, such that, if the chemical potential is fixed at the value
$\m = -2t_{2} \cos\phi \alpha_{1}(k_{F}^{\o}) + \xi(U,W, \phi)$, then, for all the values of $W,\phi$ such that $m_\o^R(W,\phi):=W+\o3\sqrt3 t_2\sin\phi+\d_\o(U,W,\phi)$ is different from zero, both for 
$\o=+$ and for $\o=-$, the interacting Hall conductivity is 
\be
\s_{12}(U)= \frac{1}{2\pi}\big[ \sign(m^{\text{R}}_{+}) - \sign(m^{\text{R}}_{-}) \big]\;.\label{hh}
\ee
Moreover, the conditions $m_\o^R(W,\phi)=0$, $\o\in\{\pm\}$, define two $C^1$ curves $W=W^R_\o(\phi)$, called `critical curves', 
which are $C^1$ close to the unperturbed curves $W=-\o3\sqrt3 t_2\sin\phi$.
The two critical curves have the same qualitative properties as the unperturbed ones, in the sense that: 
(i) they intersect at $(W,\phi)=(0,0), (0,\pi)$; (ii) they are one the image of the other, under 
the reflection $W\to-W$; (iii) they are monotone for $\phi\in[-\frac\pi2,\frac\pi2]$; (iv) they are odd in $\phi$, and their periodic extension to $\mathbb R$ is 
even under the reflection $\phi\to \pi-\phi$. 
\end{thm}
The main improvement of this result with respect to previous works is that it establishes the quantization of the Hall conductivity for values of the coupling constant $U$ that are {\it much larger} than the gap of the bare Hamiltonian: it states that the interaction does not change the value of the interacting Hall conductivity, provided we do not cross the interacting critical curves, which we construct explicitly; this universality of the Hall coefficient holds, in particular, arbitrarily close to the critical curves. On the critical curves the system is massless, i.e., correlations decay algebraically at large distances, and we do not have informations on the transverse conductivity coefficient. However, the critical 
longitudinal conductivity displays the same quantization phenomenon as the non-interacting one: namely, if $W=W^{R}_\o(\phi)$, for either $\o=+$ or $\o=-$, and $\phi\neq0,\pi$, then
\be\label{eq:univ11c}
\s_{11} = \s_{22} = \frac{1}{8}\;,
\ee
while $\s_{11} =\s_{22}= \frac{1}{4}$ for $(W,\phi)=(0,0),(0,\pi)$; see \cite{GJMP} for the proof. 

Finally, let us stress that we focus on a specific class of many-body perturbations of the Haldane model just for the sake of definiteness: the method of the proof actually 
applies to {\it any} translation-invariant interacting Hamiltonian of the form $\mathcal{H} = \mathcal{H}_0 + U\mathcal{V}$, with: (i) $\mathcal{V}$ a 
short-range, spin-independent, 
interaction, (ii) $|U|$ small compared to the bandwidth, and (iii) 
$\mathcal{H}_0$ a quadratic Hamiltonian that can become gapless as a parameter is varied: in the gapless case, $\mathcal H_0$ has a 
degenerate, point-like, Fermi surface, around which the dispersion relation has a linear, `graphene-like', behavior.

\subsubsection{Strategy of the proof}\label{sec.str}

Let us give an informal summary of the main steps of the proof. For simplicity, we limit ourselves to the generic case $W\neq0$, $\phi\neq 0$,  the special, symmetric, complementary case ($W=0$ and/or $\phi=0$) 
being treatable analogously. Thanks to the symmetries of the model, see Eqs.\eqref{eq.4.6}--\eqref{eq.4.12} below, we further restrict ourselves, without loss of generality, to the range of parameters 
\be \label{eq.parrange} W>0,\qquad  0<\phi\le \frac\pi2,\ee
which corresponds to the case $m_{+}>|m_{-}|$, where $m_\pm$ are defined in \eqref{eq.mpm}. Note that, under these conditions, the amplitude of the bare gap is given by $|m_-|$. 

We expect the interaction to modify (`renormalize') in a non trivial way both the chemical potential and the width of the gap\footnote{Here, by `gap' we mean the rate of the exponential decay of the Euclidean correlations.}. In order to compute the interacting gap, we proceed as follows. For the purpose of this discussion, let us denote by $\mathcal H_0(W,\phi,\mu)$ the non-interacting Hamiltonian \eqref{eq:Haldane}, thought of as a function of the parameters $(W,\phi,\mu)$, at fixed $t_1,t_2$. We rewrite $\mu$ in the form $\mu=-2t_{2}\cos\phi\, \alpha_{1}(k_{F}^{\o})-\xi$, and, recalling that $W=m_-+3\sqrt3 t_2\sin\phi$, we rewrite 
$W=(m_--\d)+3\sqrt3t_2\sin\phi +\d \equiv m_{\text{R},-}+3\sqrt3t_2\sin\phi +\d$, where $\d$ will be chosen in such a way that $m_{\text{R},-}=m_--\d$ has the interpretation of {\it renormalized gap}. 
By using these rewritings, we find:
$$\mathcal H=\mathcal H_0(W,\phi,\mu)+U\mathcal V=\mathcal H_0^{\text{R}}(m_{\text{R},-},\phi)+
U\mathcal V+ \d\sum_{\vec x\in \L_{L}}[n_{\vec x, A} - n_{\vec x, B}] +\x \sum_{\vec x\in \L_L}n_{\vec x},$$
where 
$$\mathcal H_0^{\text{R}}(m_{\text{R},-},\phi):=\mathcal H_0(m_{\text{R},-}+3\sqrt3 t_2\sin\phi, \phi,-2t_{2}\cos\phi\, \alpha_{1}(k_{F}^{\o}))$$
and $n_{\vec x,\rho}=\sum_{\s=\uparrow,\downarrow}n_{\vec x,\rho,\s}$. Let us now introduce the reference Hamiltonian $\mathcal H^{\text{R}}$, thought of as a function of the parameters $U,m_{\text{R},-},\phi$, defined by 
\be
\mathcal{H}^\text{R}: =\mathcal{H}_0^{\text{R}}(m_{\text{R},-},\phi) + U\mathcal{V}+ \d(U,m_{\text{R},-},\phi)\sum_{\vec x\in \L_{L}}[n_{\vec x, A} - n_{\vec x, B}] +\x(U,m_{\text{R},-},\phi) 
\sum_{\vec x\in \L_L}n_{\vec x}.\label{cc}\ee
In general, $\mathcal{H}^\text{R}$ is different from the original Hamiltonian $\mathcal H$. 
However, by construction, $\mathcal H=\mathcal H^{\text{R}}$, provided that 
$\mu=-2t_{2}\cos\phi\, \alpha_{1}(k_{F}^{\o})-\xi(U,m_{\text{R},-},\phi)$, and $m_{\text{R},-}$ is a solution of the fixed point equation
\be
m_{\text{R}, -} =W-3\sqrt3 t_2\sin\phi -\d(U,m_{\text{R},-},\phi)\;.\label{eq.mR-}
\ee
Our construction, described below, will allow us to fix the counterterms $\xi(U,m_{\text{R},-},\phi)$ and $\d(U,m_{\text{R},-},\phi)$ in such a way that they are small, of order $O(U)$, and 
that, as anticipated above, $m_{\text{R},-}$ has the interpretation of renormalized gap: in particular, the condition  
$m_{\text{R},-}\neq0$ implies that the system is massive, that is, correlations decay exponentially at large distances, with decay rate $m_{\text{R},-}$. 

\medskip

Given these definitions, the main steps of the proof are the following. 
\begin{enumerate}
\item We introduce the Euclidean correlations and the Euclidean Hall conductivity, which are formally obtained from the corresponding real-time formulas 
via a `Wick rotation' of the time variable. In Lemma \ref{lem:univ}, 
by differentiating the Ward Identities associated with the continuity equation, 
and by combining the result with the Schwinger-Dyson equation, we show that the Euclidean Hall conductivity of $\mathcal H^{\text R}$ is constant in $U$,
provided that $\x(U,m_{\text{R},-},\phi),\d(U,m_{\text{R},-},\phi)$ are differentiable in $U$ and that the Fourier transform of the Euclidean correlation functions is smooth (i.e., at least of class $C^{3}$)
in the momenta, for any fixed $m_{\text{R},-}\neq0$. 
\item As a second step, we prove the assumptions of Lemma \ref{lem:univ}. More precisely, 
we prove that there exist two functions $\x(U, m_{\text{R}, -}, \phi)$ and $\d(U, m_{\text{R}, -}, \phi)$, analytic in $U$, such that
the Euclidean correlations of the model (\ref{cc}) are analytic in $U$ and, if $m_{\text{R},-}\neq0$, they are exponentially decaying at large space-time distances, with decay rate $m_{\text{R},-}$; in particular, if $m_{\text{R},-} \neq0$, their Fourier transform is smooth in the momenta. 
\item Next, we prove the equivalence between the original model and the model with 
Hamiltonian $\mathcal H^{\text{R}}$, anticipated above. In particular, we prove that $\d$ is differentiable in $m_{\text{R}, -}$, with small (i.e., $O(U)$) derivative;  
therefore, eq.\eqref{eq.mR-} can be solved via the implicit function theorem, thus giving 
\be
m_{\text{R},-} =W-3\sqrt3 t_2\sin\phi + \widetilde \d(U,W,\phi), \label{eq.2.23}\ee
and we show that $|\widetilde\d(U,W,\phi)|\leq C|U| (W+\sin\phi)$. 
The equation for the interacting critical curve has the form: $W=3\sqrt{3}\,t_{2}\sin\phi+\d(U,0,\phi)=(1+O(U))\,3\sqrt3 t_2\sin\phi$. 
\item Finally, once we derived explicit estimates on the decay properties of the Euclidean correlations, we infer the 
identity between the Euclidean and the real-time Kubo conductivity, via \cite[Lemma B.1]{AMP}.
\end{enumerate}

The key technical difference with respect to the strategy in \cite{GMPhall} is the rewriting of the model in terms of the renormalized reference 
Hamiltonian $\mathcal{H}^{\text{R}}_0$: this allows us to take into account the renormalization of the gap and of the chemical potential, which characterizes the 
interacting critical point of the theory.

\section{Lattice conservation laws and universality}\label{sec:proof}

In this section, we show how lattice conservation laws can be used to prove the universality of the Euclidean Kubo conductivity, see step (i) above. 
The main result of this section is summarized in Lemma \ref{lem:univ}. Before getting to this lemma, in Section \ref{sec:Eu} we introduce the Euclidean formalism
and derive the {\it Ward identities}, associated with the lattice continuity equation (\ref{eq:cont}), for the Euclidean correlations. In Sections \ref{sec:cons1} and \ref{sec:cons2} we differentiate and manipulate the Ward identities, under the assumption that the current-current correlations are sufficiently smooth in momentum space, 
thus getting some important identities, summarized in Lemma \ref{lem:C1} and \ref{lem:C2}. Finally, in Section \ref{sec:univ}, we prove Lemma \ref{lem:univ}, 
by combining these identities with the Schwinger-Dyson equation. 

\subsection{Euclidean formalism and Ward identities}\label{sec:Eu}

Given an operator $\mathcal{O}$ on $\mathcal{F}_{L}$ and $t\in [0, \beta)$, we define the imaginary-time evolution generated by the Hamiltonian $\mathcal{H}^\text{R}$, Eq. (\ref{cc}), 
as: $\mathcal{O}_{t} := e^{t \mathcal{H}^{\text{R}}} \mathcal{O} e^{-t\mathcal{H}^{\text{R}}}$. Notice that $\mathcal{O}_{t} \equiv \mathcal{O}(-it)$, with $\mathcal{O}(t)$ the real-time 
evolution generated by $\mathcal{H}^{\text{R}}$. Given $n$ operators $\mathcal{O}_{t_{1}}^{(1)},\ldots, \mathcal{O}_{t_{n}}^{(n)}$ on $\mathcal{F}_{L}$, each of which (i) can be written as 
a polynomial in the time-evolved creation and annihilation operators $\psi^\pm_{(t,\vec x),\rho} = e^{t \mathcal{H}^\text{R}} \psi^\pm_{ \vec x,\rho}e^{-t\mathcal{H}^\text{R}}$, (ii) is 
normal-ordered, and (iii) is either even or odd in $\psi^\pm_{(t,\vec x),\rho}$, we define their time-ordered average, or {\it Euclidean correlation function}, as:
\begin{equation}\label{eq:18}
\langle {\bf T}\, \mathcal{O}^{(1)}_{t_1}\cdots \mathcal{O}^{(n)}_{t_{n}} \rangle_{\beta,L}^{\text{R}} := \frac{\Tr_{\mathcal{F}_{L}} e^{-\beta \mathcal{H}^\text{R}} \mathbf{T} \big\{ \mathcal{O}_{t_{1}}^{(1)}\cdots \mathcal{O}_{t_{n}}^{(n)} \big\}  }{\Tr_{\mathcal{F}_{L}} e^{-\beta \mathcal{H}^{\text{R}}}}
\;,
\end{equation}
where the (linear) operator $\mathbf{T}$ is the fermionic time-ordering, acting
on a product of fermionic operators as:
\begin{equation}
\mathbf{T} \big\{ \psi^{\e_{1}}_{(t_1,\vec x_1),\s_1}\cdots \psi^{\e_{n}}_{(t_n,\vec x_n),\s_n} \big\} = \text{sgn}(\pi) \psi^{\e_{\pi(1)}}_{(t_{\p(1)},\vec x_{\p(1)}),\s_{\p(1)}}\cdots \psi^{\e_{\pi(n)}}_{
(t_{\p(n)},\vec x_{\p(n)}),\s_{\p(n)}}
\;,
\end{equation}
where $\pi$ is a permutation of $\{1,\ldots, n\}$ with signature $\text{sgn}(\pi)$
such that $t_{\pi(1)}\geq \ldots \geq t_{\pi(n)}$. If some operators are evaluated at the same time, 
the ambiguity is solved by normal ordering. 
We also denote the {\it connected} Euclidean correlation function, or cumulant, by $\langle {\bf T}\, \mathcal{O}^{(1)}_{t_1}\,; \mathcal{O}^{(2)}_{t_2}\,; \cdots \,; \mathcal{O}^{(n)}_{t_{n}} 
\rangle_{\beta,L}^{\text{R}}$.

Let $O$ be a self-adjoint operator on $\mathcal{F}_{L}$. We define its time Fourier transform as: $\widehat{\mathcal{O}}_{p_0} = \int_{0}^{\beta} dt\, e^{-ip_0 t}
\mathcal{O}_{t}$ with $p_0 \in \frac{2\pi}{\beta}\mathbb{Z}$ the {\it Matsubara frequencies}. Also, we denote by $\widehat{\mathcal{O}}_{\pp}$, for 
$\pp =  (p_{0}, p_{1}, p_{2})$, the joint space-time Fourier transform of the operator $\mathcal{O}_{(t,\vec x)}$: $$\widehat{\mathcal{O}}_{\pp} = 
\int_{0}^{\beta} dt\, \sum_{\vec x\in \L_{L}} e^{-i\pp\cdot \xx} \mathcal{O}_{\xx},$$ with 
$\xx = (t, x_{1}, x_{2})\equiv (x_{0}, x_{1}, x_{2})$. 

Let $j_{\m, \vec x}$, with $\m\in\{0,1,2\}$, be the three-component operator such that $j_{0,\vec x}:=n_{\vec x}$, while $j_{i,\vec x}$, with $i\in\{1,2\}$, are the 
components of the 
total current along the reference, orthogonal, coordinate directions, see \eqref{j12}. Note that $j_{\m,\vec x}$ is the natural current operator, associated both with $\mathcal H$ and 
with $\mathcal H^{\text{R}}$, because $i[\mathcal H,n_{\vec x}]=i[\mathcal H^{\text{R}},n_{\vec x}]$. 
We define the normalized current-current correlation functions as:
\be\label{jeff}
\widehat{K}_{\m_{1}, \ldots, \m_{n}}^{\b, L; \text{R}}(\pp_{1}, \ldots, \pp_{n-1}) := \frac{1}{\beta L^2} \langle {\bf T}\, \hat{\jmath}_{\m_1, \pp_{1}}\,; \hat{\jmath}_{\m_{2}, \pp_{2}}\,; \cdots \,; \hat{\jmath}_{\m_{n}, -\pp_{1}-\ldots - \pp_{n-1}} \rangle_{\b, L}^{\text{R}}\ee
for $\m_i\in\{0,1,2\}$. We also denote the infinite volume, zero temperature limit of the Euclidean correlations by: $\widehat{K}^{\text{R}}_{\m_{1}, \ldots, \m_{n}}(\pp_{1}, \ldots, 
\pp_{n-1}) := \lim_{\b\to \infty}\lim_{L\to \infty} \widehat{K}_{\m_{1}, \ldots, \m_{n}}^{\b, L; \text{R}}(\pp_{1}, \ldots, \pp_{n-1})$.
The {\it Euclidean conductivity matrix} for $\mathcal H^{\text{R}}$ is
\be\label{eq:sijEu}
\s_{ij}^{\text{E}, \text{R}} :=\frac1{|\vec\ell_1\wedge\vec\ell_2|} \lim_{p_{0}\to 0^{+}}\frac{1}{p_{0}} \Big(- \widehat{K}^{\text{R}}_{i,j}\big((-p_0,\vec 0)\big) 
+i \pmb{\langle}  [ \mathcal{J}_{i}, \mathcal{X}_{j} ] \pmb{\rangle}_{\infty}^{\text{R}}\Big)\;,\ee
where, in the second term, $\pmb{\langle}\cdot \pmb{\rangle}_{\infty}^R:=\lim_{\b,L\to\infty}\frac{1}{L^{2}} \langle \cdot\rangle_{\b, L}^{\text{R}}$, and the 
expression $[ \mathcal{J}_{j}, \mathcal{X}_{i} ]$ must be understood as explained in the footnote \ref{ciaciao} above. 
This definition can be obtained via a formal `Wick rotation' of the time variable, $t\to -it$, starting from the original definition of the Kubo conductivity, \eqref{eq:realkubo}, see, e.g., \cite{Ale}. 
A posteriori, we will see that in our context the two definitions coincide, see Section \ref{sec.5} below. 

The structure correlation functions, and hence the conductivity, is severely constrained by {\it lattice Ward identities}. These are nonperturbative implications of lattice continuity equation Eq.(\ref{eq.2.11}), which we rewrite here in imaginary time:
\be\label{eq:consI}
i\partial_{x_{0}} j_{0,\xx} + \text{div}_{\vec x}\vec \jmath_{\xx}=0\;,
\ee
where we used the notation $\text{div}_{\vec x}\vec \jmath_{\xx}:=\sum_{i=1,2}\text{d}_{i}\tilde{\text{\j}}_{i,\xx}$. 

For instance, consider the current-current correlation function\footnote{The definition in \eqref{eq.to} is only valid for $x_0\neq y_0$. However, the value at $x_0=y_0$ has no influence on the following formulas, in particular on \eqref{eq:WIf}, which is the main goal of the following manipulations.}, 
\be\label{eq.to}
\langle {\bf T}\, j_{0, \xx}\,; j_{\n, \yy} \rangle^{\text{R}}_{\b,L} = \theta(x_{0} - y_{0}) \langle  j_{0, \xx}\,; j_{\n, \yy} \rangle^{\text{R}}_{\b,L} + \theta(y_{0} - x_{0}) \langle j_{\n, \yy}\,; j_{0, \xx} \rangle^{\text{R}}_{\b,L}\;,
\ee
where  $\theta(t)$ is the Heaviside step function and the correlations in the right side are the time-unordered ones (i.e., 
they are defined without the action of the time-ordering operator).
Using the continuity equation Eq. (\ref{eq:consI}):
\bea
i\partial_{x_{0}}\langle {\bf T}\, j_{0, \xx}\,; j_{\n, \yy} \rangle^{\text{R}}_{\b,L} &=& \langle {\bf T}\,  i\partial_{x_{0}}j_{0, \xx}\,; j_{\n, \yy} \rangle^{\text{R}}_{\b,L} + i\langle [ j_{0, \vec x}\, , j_{\n, \vec y} ] \rangle^{\text{R}}_{\b,L} \d(x_{0} - y_{0})\nn\\
&=& -\langle {\bf T}\,  \text{div}_{\vec x} \vec \jmath_{\xx}\,; j_{\n, \yy} \rangle^{\text{R}}_{\b,L} + i\langle [ j_{0, \vec x}\, , j_{\n, \vec y} ] \rangle^{\text{R}}_{\b,L} \d(x_{0} - y_{0})\;.
\label{eq.ccc}\eea
Let us now take the Fourier transform of both sides: integrating by parts w.r.t. $x_0$ and using \eqref{eq.ccc}, we find
\bea
p_{0} \widehat{K}^{\b, L; \text{R}}_{0,\n}(\pp) &=& -\frac{1}{\beta L^2} \int_{0}^{\beta} dx_{0} \int_{0}^{\beta} dy_{0}\,  \sum_{\vec x, \vec y\in\L_L} e^{-ip_{0}(x_{0} - y_{0})}e^{-i\vec p
\cdot  (\vec x-\vec y)} i\partial_{x_{0}}\langle {\bf T}\, j_{0, \xx}\,; j_{\n, \yy} \rangle^{\text{R}}_{\b,L}\nn\\
&=& \sum_{i=1,2} (1 - e^{-i\vec p\cdot\vec\ell_{i}}) \frac{1}{\beta L^{2}} \langle {\bf T}\, \hat{\vec \jmath}_{\pp}\cdot\frac{{\vec G_i}}{2\pi}\,; \hat \jmath_{\n, -\pp} \rangle^{\text{R}}_{\b,L} - i\sum_{\vec x} e^{-i\vec p\cdot \vec x} \langle [ j_{0, \vec x}\, ,  j_{\n, \vec 0} ] \rangle^{\text{R}}_{\b,L}\nn\\
&\equiv& \sum_{i,i'=1,2} (1 - e^{-i\vec p\cdot\vec\ell_i})\frac{(\vec G_i)_{i'}}{2\pi} \widehat{K}^{\b, L; \text{R}}_{i',\n}(\pp) + \widehat S^{\b, L; \text{R}}_{\n}(\pp)\;,
\label{eq:WIf}\eea
where we used that $\tilde{\text{\j}}_{i,\xx}=\vec\jmath_{\xx}\cdot\frac{\vec G_i}{2\pi}$, with $\vec G_i$, $i=1,2$, the vectors of the dual basis, see definition in Section \ref{sec.2.1}.
More generally, denoting $(0,\n_2,\ldots,\n_{n})$ by $(0,\underline\nu)$, one has:
\bea\label{eq:WIgen}
&& p_{1,0} \widehat{K}^{\b, L; \text{R}}_{0, \underline{\n}}(\{\pp_{i}\}_{i=1}^{n-1}) = \sum_{i,i'=1,2} (1 - e^{-i\vec p_{1}\cdot\vec\ell_i}) \frac{(\vec G_i)_{i'}}{2\pi}
\widehat{K}^{\b, L; \text{R}}_{i', \underline{\n}}(\{\pp_{i}\}_{i=1}^{n-1}) + \widehat S^{\b, L; \text{R}}_{\underline\nu}(\{\pp_{i}\}_{i=1}^{n-1})\;,\qquad\phantom{+}\\
&&\widehat S_{\underline{\n}}^{\b, L; \text{R}}(\cdots) := -\frac{i}{\b L^2}\sum_{j=2}^{n}  \langle {\bf T}\,C_{\n_{j}}(\pp_{1}, \pp_{j})\,; \hat \jmath_{\n_{2},\pp_{2}}\,;\ldots\,; \hat \jmath_{\n_{j-1},\pp_{j-1}}\,; \hat \jmath_{\n_{j+1},\pp_{j+1}}\,; \cdots\,; \hat\jmath_{\n_{n},\pp_{n}} \rangle_{\b, L}^{\text{R}}, \nn
\eea
with $C_{\n}(\pp_{1}, \pp_{2}) = \int_{0}^{\beta} dt\, e^{-it (\o_{1} + \o_{2})} [ \hat \jmath_{0,(t, \vec p_{1})}\,, \hat \jmath_{\n, (t, \vec p_{2})}]$ (here, with some 
abuse of notation, we let $\hat\jmath_{\mu,(t,\vec p)}$ be the imaginary-time evolution at time $t$ of $\hat\jmath_{\mu,\vec p}$), 
and with the understanding that $\pp_{n} = -\pp_{1} - \ldots - \pp_{n-1}$. Even more generally, the identity remains valid if some of the current operators $j_{\n_i,\pp_i}$ are replaced by other local operators $\hat{\mathcal O}_{i,\pp_i}$: in this case, of course, the operators  $C_{\nu_i}$ must be modified accordingly. In the following, we will be interested in replacing one of the current operators either by the staggered density 
\be \label{j3} \hat \jmath_{3,\pp}:=n_{\pp,A}-n_{\pp,B}\,,\ee
where $n_{\pp,\rho}$ is the Fourier transform of $n_{(t,\vec x),\rho}:=\sum_\s\psi^+_{(t,\vec x),\rho,\s}\psi^-_{(t,\vec x),\rho,\s}$, 
or by the quartic interaction potential
\be\label{Vp}
\hat {\mathcal{V}}_{\pp} := \int_{0}^{\beta} dx_{0}\, e^{-ip_{0} x_{0}} \sum_{\vec x} e^{-i\vec p\cdot \vec x} \sum_{\vec y, \r, \r'} v_{\r,\r'}(\vec x - \vec y)
\big( (n_{\vec x, \r}-1) (n_{\vec y, \r'}-1)\big)_{x_{0}}\;.\ee

\medskip

As we shall see below, the combination of the identity (\ref{eq:WIgen}) together with the regularity of the correlation functions has remarkable implications on the structure of the correlations.

\subsubsection{Consequences of the Ward identities for $C^1$ correlations}\label{sec:cons1}

Here we start by discussing the consequences of the Ward identities for continuously differentiable correlations.
\begin{lemma}\label{lem:C1} Let $\pp_{\b, L} \in \frac{2\pi}{\beta}\mathbb{Z} \times \frac{2\pi}{L} \mathbb{Z}^{2}$, such that $\lim_{\beta, L\to \infty}\pp_{\b, L} = \pp \in B_{\e}(\V0) := \{\qq\in \mathbb{R}^{2} \mid |\qq|<\e \}$, for some $\e>0$. Suppose that $\widehat K^{\text{R}}_{\m,\n}(\pp) = \lim_{\beta,L\to \infty} \widehat K^{\b, L; \text{R}}_{\m,\n}(\pp_{\b, L})$ 
and $\widehat S^{\text{R}}_{j}(\pp) = \lim_{\beta,L\to \infty} \widehat S^{\b, L; \text{R}}_{j}(\pp_{\b, L})$ exist and that $\widehat K^{\text{R}}_{\m,\n}(\pp), \widehat{S}^{\text{R}}_{j}(\pp)  \in C^{1}(B_{\e}(\V0))$.
Then,
\be\label{jeff1}
\s^{\text{E}, \text{R}}_{ij}= \frac1{|\vec\ell_1\wedge\vec\ell_2|}
\frac{\partial}{\partial p_{0}} \widehat{K}^{\text{R}}_{i,j}({\bf 0})\;.
\ee
\end{lemma}
\begin{proof} Consider Eq. (\ref{eq:WIf}) with $\n = j$, in the $\beta, L\to \infty$ limit. We differentiate both sides w.r.t. $p_{i}$, and take the limit $\pp\to \V0$, thus getting (recall that $\vec\ell_i\cdot\vec G_j=2\pi\d_{i,j}$):
\be
0 = i \widehat{K}^{\text{R}}_{i,j}(\V0) + \frac{\partial}{\partial p_{i}} \widehat S^{\text{R}}_{j}(\V0).\ee
Now, recall the definition of $\widehat S^{\b,L;\text{R}}_{j}(\pp)$ from eq.\eqref{eq:WIf}: $\widehat S^{\b,L;\text{R}}_{j}(\V0)=- i\sum_{\vec x} e^{-i\vec p\cdot \vec x} \langle [n_{\vec x}\, ,  j_{j, \vec 0} ] \rangle^{\text{R}}_{\b,L}$, where we also used that $j_{0,\vec x}=n_{\vec x}$. Taking the limit $\b,L\to\infty$ and the derivative with respect to $p_i$, we get 
$ \frac{\partial}{\partial p_{i}} \widehat S^{\text{R}}_{j}(\V0)=-\pmb{\langle} [\mathcal X_{i},\mathcal J_j ]  \pmb{\rangle}^{\text{R}}_{\infty}$, where
$\pmb{\langle}\cdot \pmb{\rangle}_{\infty}^R$ was defined in \eqref{eq:sijEu}, and the 
expression $[ \mathcal{J}_{j}, \mathcal{X}_{i} ]$ must be understood as explained in the footnote \ref{ciaciao} above. In conclusion, 
\be  \widehat{K}^{\text{R}}_{i,j}(\V0) =i\pmb{\langle} [\mathcal J_j, \mathcal X_{i} ]  \pmb{\rangle}^{\text{R}}_{\infty}\;,
\ee
and, if we plug this identity in \eqref{eq:sijEu}, noting that 
$\pmb{\langle} [\mathcal J_j, \mathcal X_{i} ]  \pmb{\rangle}^{\text{R}}_{\infty}$ is even under the exchange $i\otto j$, we obtain 
the desired identity. 
\end{proof}

\subsubsection{Consequences of the Ward identities for $C^{3}$ correlations}\label{sec:cons2}

Next, we discuss some other implications of the Ward identities for $C^3$ three-point correlations of the current operator (twice) with either the staggered density 
$\hat j_{3,\pp}$ (see \eqref{j3}), or the interaction potential (see \eqref{Vp}), defined as
\bea
\widehat{K}^{\beta, L; \text{R}}_{\m, \n, 3}(\pp, \qq) &:=& \frac1{\b L^2}\langle {\bf T}\,\hat \jmath_{\mu, \pp}\,; \hat \jmath_{\nu,\qq}\,; \hat \jmath_{3,-\pp-\qq} {\rangle}_{\beta, L} \nonumber\\
\widehat{K}^{\beta, L; \text{R}}_{\m, \n, V}(\pp, \qq) &:=& \frac1{\b L^2}\langle {\bf T}\,\hat \jmath_{\mu, \pp}\,; \hat \jmath_{\nu,\qq}\,; \hat{\mathcal{V}}_{-\pp-\qq} {\rangle}_{\beta, L} 
\;.
\label{eq.pot}\eea
We also let %
\bea\label{eq:schwaleph}
\widehat{S}^{\beta, L; \text{R}}_{j, 3}(\pp, \qq) &:=& -\frac{i}{\b L^2}\langle C_{j}(\pp, \qq)\,; \hat \jmath_{3,-\pp-\qq}\rangle_{\beta, L}^{\text{R}}\;,\\
\widehat{S}^{\beta, L; \text{R}}_{j, V}(\pp, \qq) &:= &-\frac{i}{\b L^2}\langle C_{j}(\pp, \qq)\,; \hat{\mathcal{V}}_{-\pp-\qq}  {\rangle}_{\beta, L}^{\text{R}}\;.
\eea
be the new Schwinger terms. As usual, we denote by 
$\widehat{K}^{\text{R}}_{\m, \n, \sharp}$, $\widehat{S}^{\text{R}}_{j, \sharp}$ the $\beta, L\to\infty$ limits of $\widehat{K}^{\beta, L; \text{R}}_{\m, \n, \sharp}(\cdots)$, $\widehat{S}^{\beta, L; \text{R}}_{j, \sharp}(\cdots)$, with $\sharp\in\{3,V\}$.
\begin{lemma}\label{lem:C2} Let $\sharp \in\{0,3,V\}$. Suppose that the limiting functions $\widehat{K}^{\text{R}}_{\m, \n, \sharp}(\pp,\qq)$, $\widehat{S}^{\text{R}}_{j, \sharp}(\pp,\qq)$
exist in $B_{\e}(\V0)\times B_\e(\V0)$, and that they are of class $C^{3}$ in this domain. Then:
\be\label{eq:WI3}
\frac{\partial}{\partial p_{0}} \widehat{K}^{\text{R}}_{i, j, \sharp}((p_{0}, \vec 0), (-p_{0}, \vec 0)) = \frac{\partial}{\partial p_{0}}\Big[ 
p_{0}^2 \frac{\partial^{2}}{\partial p_{i} \partial q_{j}} \widehat{K}^{\text{R}}_{0,0, \sharp}((p_{0}, \vec 0), (-p_{0}, \vec 0))\Big]\;.
\ee
In particular, the left side of Eq. (\ref{eq:WI3}) vanishes as $p_{0}\to 0$.
\end{lemma}

\begin{proof} Taking the $\b,L\to\infty$ limit of the Ward Identity \eqref{eq:WIgen} with $\underline\nu=(0,\sharp)$, we find
\be
p_{0} \widehat{K}^{\text{R}}_{0,0, \sharp}(\pp, \qq) = \sum_{i,i'=1,2} (1 - e^{-i\vec p\cdot\vec \ell_{i}})\frac{(\vec G_i)_{i'}}{2\pi} \widehat{K}^{\text{R}}_{i', 0, \sharp}(\pp, \qq)\;.\label{3.21}\ee
Similarly, choosing $\underline\nu=(j,\sharp)$
\be
p_{0} \widehat{K}^{\text{R}}_{0,j, \sharp}(\pp, \qq) = \sum_{i,i'=1,2} (1 - e^{-i\vec p\cdot\vec\ell_{i}})\frac{(\vec G_i)_{i'}}{2\pi} \widehat{K}^{\text{R}}_{i', j, \sharp}(\pp, \qq) + \widehat S^{\text{R}}_{j, \sharp}(\pp,\qq)\;,
\ee
and, exchanging the roles of $\pp$ and $\qq$, we also get
\be
q_{0} \widehat{K}^{\text{R}}_{i,0, \sharp}(\pp, \qq) = \sum_{j,j'=1,2} (1 - e^{-i\vec q\cdot\vec\ell_{j}})\frac{(\vec G_j)_{j'}}{2\pi} \widehat{K}^{\text{R}}_{i, j', \sharp}(\pp, \qq) + 
\widehat S^{\text{R}}_{i, \sharp}(\pp,\qq)\;.\label{3.23}
\ee
Combining \eqref{3.21} with \eqref{3.23}, we find
\bea
q_{0} p_{0} \widehat{K}^{\text{R}}_{0,0, \sharp}(\pp, \qq) =\sum_{i,i' = 1,2}&\Big[& \sum_{j,j' = 1,2}(1 - e^{-i\vec p\cdot\vec\ell_{i}}) (1 - e^{-i\vec q\cdot\vec\ell_{j}})
\frac{(\vec G_i)_{i'}}{2\pi}\frac{(\vec G_j)_{j'}}{2\pi}
\widehat{K}^{\text{R}}_{i', j', \sharp}(\pp, \qq) \nn\\
&+& (1-e^{-i\vec p\cdot\vec \ell_{i}})\frac{(\vec G_i)_{i'}}{2\pi}\widehat S^{\text{R}}_{i', \sharp}(\pp,\qq)\ \ \Big]\;.
\eea
We now derive w.r.t. $p_{i}, q_{j}$, and then set $\pp = -\qq = (p_{0}, \vec 0)$, thus finding:
\be
p_{0}^{2} \frac{\partial^{2}}{\partial p_{i} \partial q_{j}} \widehat{K}^{\text{R}}_{0,0, \sharp}\big((p_0,\vec 0),(-p_0,\vec 0)\big) = \widehat{K}^{\text{R}}_{i, j, \sharp}\big((p_{0}, \vec 0), (-p_{0}, \vec 0)\big) -i  \frac{\partial}{\partial q_{j}}\widehat S^{\text{R}}_{i, \sharp}\big((p_0,\vec 0),(-p_0,\vec 0)\big)\;.
\ee
Finally, notice that $\partial_{q_{j}} S^{\text{R}}_{i, \sharp}\big((p_0,\vec 0),(-p_0,\vec 0)\big)$ is constant in $p_{0}$ (recall the definition of Schwinger term, Eq. (\ref{eq:schwaleph}), and of $C_{j}$, Eq. (\ref{eq:WIgen})). Therefore, after differentiation in $p_{0}$, the final claim follows.
\end{proof}

\subsection{Universality of the Euclidean conductivity matrix}\label{sec:univ}

Here we prove the universality of the Euclidean conductivity matrix, defined in Eq. (\ref{eq:sijEu}). We restrict to the range of parameters \eqref{eq.parrange}, as discussed at the 
beginning of Section \ref{sec.str}. In terms of the renormalized paramters, we restate \eqref{eq.parrange} as
\be \label{eq.parrange1} 0<\phi\le \frac{\pi}2\;,\qquad m_{\text{R},+}>|m_{\text{R},-}|\;,\ee
where 
\be m_{R,+}:=m_{R,-}+6\sqrt{3}\,t_{2}\sin\phi\;.\label{eq.mR+}\ee
A key ingredient in the proof is the following regularity result for the correlation functions.

\begin{prop}\label{assm} There exists $U_{0}>0$ such that, for $|U|<U_{0}$ and for parameters $(\phi,m_{\text{R},-})$ in the range \eqref{eq.parrange1},  
the following is true. There exist functions $\x(U, m_{\text{R}, -}, \phi)$, $\d(U, m_{\text{R}, -}, \phi)$, analytic in $U$ and vanishing at $U=0$, such that the Euclidean correlation 
functions $\widehat K^{\text{R}}_{\m,\n}(\pp)$, $\widehat K^{\text{R}}_{\m, \n, \sharp}(\pp, \qq)$, as well as the Schwinger terms $\widehat{S}^{\text{R}}_{j}(\pp)$, 
$\widehat{S}^{\text{R}}_{j, \sharp}(\pp,\qq)$, with $\sharp\in\{0,3,V\}$, are analytic in $U$; moreover, if $m_{\text{R},-}\neq 0$, 
they are $C^{3}$ in $\pp, \qq\in B_{\e}(\V0)$, uniformly in $U$ and $\phi$.
\end{prop}

The proof of this proposition is postponed to the next section. Its content, combined with the (consequences of the) Ward identities discussed above, 
immediately implies the universality of the Euclidean conductivity matrix.

\begin{lemma}\label{lem:univ} Under the same assumptions as Proposition \ref{assm}, if $m_{\text{R},-0}\neq 0$, then 
\be\label{eq:s12R}
\s^{\text{E},\text{R}}_{12} = \frac{1}{2\pi} \big[ \sign(m_{\text{R},+}) - \sign(m_{\text{R},-})\big]\;.
\ee 
\end{lemma}


\begin{proof} ({\it Assuming the validity of Proposition \ref{assm}}). Thanks to Proposition \ref{assm}, we know that the correlation functions $\widehat K^{\text{R}}_{\m,\n}(\pp)$, $\widehat K^{\text{R}}_{\m, \n, \sharp}(\pp, \qq)$, and the Schwinger terms $\widehat{S}^{\text{R}}_{j}(\pp)$, $\widehat{S}^{\text{R}}_{j, \sharp}(\pp,\qq)$, with $\sharp\in\{0,3,V\}$, are $C^{2}$ in $\pp, \qq\in B_{\e}(\V0)$, for $|U|< U_0$. Therefore, we can apply Lemma \ref{lem:C1} and Lemma \ref{lem:C2}. Using Lemma \ref{lem:C1}, we rewrite the Euclidean conductivity matrix as:
\be
\s^{\text{E}, \text{R}}_{ij} =\frac1{|\vec\ell_1\wedge\vec\ell_2|} \frac{\partial}{\partial p_{0}} \widehat{K}^{\text{R}}_{i,j}(\V0)\;.
\ee
Then, we rewrite $\widehat{K}^{\text{R}}_{i,j}$ in terms of the non-interacting current-current correlation associated with $\mathcal H_0^{\text{R}}$, via the 
following {\it interpolation formula}:
\be
\widehat{K}^{\text{R}}_{i,j}(\pp) = \widehat{K}^{\text{R},0}_{i,j}(\pp) + \int_{0}^{U} dU'\, \frac{d}{dU'} \widehat{K}^{\text{R},U'}_{i,j}(\pp)\;.
\ee
where $\widehat{K}^{\text{R},U'}_{i,j}(\pp)$ is the correlation associated with the ($\b,L\to\infty$ limit of the) Gibbs measure with Hamiltonian 
\be
\mathcal{H}^\text{R}_{U'}: =\mathcal{H}_0^{\text{R}} + U'\mathcal{V}+ \d(U',m_{\text{R},-},\phi)\sum_{\vec x\in \L_{L}}[n_{\vec x, A} - n_{\vec x, B}] +
\x(U',m_{\text{R},-},\phi)\sum_{\vec x\in \L_L}n_{\vec x},\label{ccbis}
\ee
cf. with Eq.\eqref{cc}. Computing the derivative in $U'$:
\bea\label{eq:int}
\widehat{K}^{\text{R}}_{i,j}(\pp) &=& \widehat{K}^{\text{R},0}_{i,j}(\pp)- \int_{0}^{U} dU'\, \Big[\frac{\partial\delta}{\partial U'}(U',m_{\text{R},-},\phi)\,  \widehat{K}^{\text{R},U'}_{i,j,3}(\pp,-\pp)\\
&+&  \frac{\partial\xi}{\partial{U'}}(U',m_{\text{R},-},\phi)\, \widehat{K}^{\text{R},U'}_{i,j,0}(\pp,-\pp) 
 + \widehat K^{\text{R}, U'}_{i,j,V}(\pp, -\pp)\Big]\;.\nn
\eea
We now take the derivative w.r.t. $p_{0}$ and take $p_0\to0$. Using Lemma \ref{lem:C2}, we immediately get:
\be
\frac{\partial}{\partial p_{0}}\widehat{K}^{\text{R}}_{i,j}(\V0) =\frac{\partial}{\partial p_{0}}\widehat{K}^{\text{R},0}_{i,j}(\V0)\;,
\ee
that is, $\s^{\text{E},\text{R}}_{ij}= \s^{\text{E}, \text{R}}_{ij}\Big|_{U=0}$ (we recall that $\s^{\text{E}, \text{R}}_{ij}\Big|_{U=0}$ is the non-interacting 
Euclidean conductivity associated with 
the quadratic Hamiltonian $\mathcal H_0^{\text{R}}$ at $m_{\text{R},-}$, which is assumed to be different from zero). 
The final claim, Eq. (\ref{eq:s12R}), follows from a direct computation of the non-interacting conductivity, cf. with \cite[Appendix B, Eq.(B.8)]{GMPhall}.
\end{proof}

\section{Proof of Proposition \ref{assm}}\label{sec.4a}



The proof of Proposition \ref{assm} is a rather standard application of RG methods for fermions (see, e.g., \cite{BM1, GeM, Gi, M1} for reviews). A similar analysis for interacting graphene,
which corresponds to the case $t_{2} = W = 0$, has been discussed in 
\cite{GM, GMP1}, which we refer to for further details. See also \cite{GJMP}, where an application to the Haldane-Hubbard model was discussed. 
The RG construction of the ground-state correlation functions, 
uniformly in the gap, is ultimately made possible by the fact that the many-body interaction, in the critical, massless, case, is {\it irrelevant} in the RG sense.
The only qualitative effect of the interaction, with respect to 
the non-interacting theory, is a finite renormalization of the gap, of the chemical potential, of the Fermi velocity and of the wave function renormalization. 

We recall once more that we restrict the discussion to the range of parameters \eqref{eq.parrange1}. Moreover, we assume that $W$ is not too large, $W\le M_0$, for a pre-fixed constant $M_0$, the case of large $W$ being substantially simpler, and left to the reader (for large $W$, the system is massive and is in a trivial, non-topological, insulating phase, as it follows from the proof of \cite{GMPhall}). Finally, for simplicity, we set $t_1=1$, that is, in we 
set the scale of the bandwidth equal to one. 

\begin{proof} The starting point is the well-known representation of the Euclidean correlation in terms of Grassmann integrals (see, for instance, \cite{GM, GMPhall}). The generating functional of the correlations is denoted by $\mathcal{W}(f,A)$, with $f$ an external Grassmann field coupled to the fermionic fields, and 
$A$ a (five-component) external complex field conjugated to the lattice currents and the quartic interaction. We have:
\be
e^{\WW(f,A)}=\frac{\int P(d\Psi)e^{-V(\Psi)+(\Psi,f)+(J,A)}}{\int P(d\Psi)e^{-V(\Psi)}},\label{eq.11}
\ee
where: $\Psi^\pm_{\xx,\s}$, with $\xx=(x_0,\vec x)\in\mathbb [0,\b)\times \L_L$ and $\s\in\{\uparrow,\downarrow\}$, 
is a two-component Grassmann spinor, whose components will be denoted by $\Psi^\pm_{\xx,\r,\s}$, with $\r=A,B$;
$P(d\Psi)$ is the fermionic Gaussian integration with propagator
\be g_{\s,\s'}(\xx,\yy)=\frac{\d_{\s,\s'}}{\b L^2}\sum_{k_0\in\frac{2\pi}{\b}(\mathbb Z+\frac12)}\ \sum_{\vec k\in\frac{2\pi}{L}\mathbb Z_L^2}e^{-i\kk(\xx-\yy)} \hat g(\kk),\label{prop}\ee
where $\mathbb Z_L=\mathbb Z/L\mathbb Z$ and, letting $$R(\vec k)= -2t_2\cos\phi\big(\a_1(\vec k)-\a_1(\vec k_F^{\pm})\big),\qquad m_{R}(\vec k)=m_{R,-}+2t_2(\a_2(k)-\a_2(k_F^-))\sin\phi,$$
and recalling that we set $t_1=1$, 
\be \hat g(\kk)=\begin{pmatrix} 
&-ik_0+R(\vec k)+ m_R(\vec k)  & - \O^*(\vec k)\\  &- \O(\vec k)
 & -i k_0+R(\vec k)- m_R(\vec k)
\end{pmatrix}^{\!\!\!-1},\nn
\ee
with the understanding that, 
at contact, $g(\xx,\xx)$ should be interpreted as $\lim_{\e\to 0^+}[g(\xx+(\e,\vec 0),\xx)\\+g(\xx-(\e,\vec 0),\xx)]$;
\bea
V(\Psi)&=&\int_{0}^\b dx_0\sum_{\vec x\in\L_L} \Big[U 
\sum_{\vec y\in\L_L}\sum_{\r, \r'=A,B} n_{\xx, \r} v_{\r,\r'}(\vec x - \vec y) n_{(x_0,\vec y), \r'}\\
&+&\d(U,m_{\text{R},-},\phi)(n_{\xx,A}-n_{\xx,B})
+\x(U,m_{\text{R},-},\phi) n_{\xx}
\Big],\nn
\eea
where $n_{\xx,\r}=\sum_{\s=\uparrow,\downarrow}\Psi^+_{\xx,\r,\s}\Psi^-_{\xx,\r,\s}$ is the Grassmann counterpart of the density operator, and
$n_{\xx}=\sum_{\r=A,B}n_{\xx,\r}$; finally, 
\bea && (\Psi,f)=\int_0^\b dx_0\sum_{\vec x\in\L_L}\sum_{\s=\uparrow\downarrow}(\Psi^+_{\xx,\s}f^-_{\xx,\s}+f^+_{\xx,\s}\Psi^-_{\xx,\s}),\nn\\
&& (J,A)=\frac1{\b L^2}\sum_{p_0\in\frac{2\pi}{\b}\mathbb Z}\ \sum_{\vec p\in\frac{2\pi}{L}\mathbb Z^2}\ \sum_{\m=0}^4\hat A_{\pp,\m}\hat J_{\pp,\m},\nn
\eea
where $\hat J_{\pp,\m}=\int_0^\b dx_0 \sum_{\vec x\in \L_L}e^{-i\pp\cdot\xx} J_{\xx,\mu}$ and: $J_{\xx,0}=n_\xx$ is the Grassmann counterpart of the density;
$J_{\xx,1},J_{\xx,2}$ are the Grassmann counterparts of the two components of the lattice current, 
\be J_{\xx,1}=\frac32(\tilde J_{\xx,1}+\tilde J_{\xx,2}),\qquad J_{\xx,2}=\frac{\sqrt3}2(-\tilde J_{\xx,1}+\tilde J_{\xx,2}),\nn\ee
 with
 $$ \tilde J_{\xx,1}= -J_{\vec x, \vec x+ \vec \ell_{1}} - J_{\vec x, \vec x+ \vec \ell_{1} - \vec \ell_{2}},\qquad 
\tilde J_{\xx,2} = -J_{\vec x, \vec x + \vec \ell_{2}} - J_{\vec x, \vec x - \vec \ell_{1} + \vec \ell_{2}},$$
and
\be J_{\vec x,\vec y} =  \sum_{\s=\uparrow,\downarrow} \big[i\Psi^{+}_{\vec y, \s} H(\vec y-\vec x) \Psi^{-}_{\vec x,\s} -i \Psi^{+}_{\vec x, \s} H(\vec x-\vec y) \Psi^{-}_{\vec y,\s}\big]\;; \nn\ee
$J_{\xx,3}=n_{\xx,A}-n_{\xx,B}$ is the Grassmann counterpart of the staggered density; 
$J_{\xx,4}$ is the Grassmann counterpart of the quartic interaction, 
\be J_{\xx,4}= \sum_{\vec y, \r, \r'} n_{\xx, \r}
v_{\r,\r'}(\vec x - \vec y)n_{(x_0,\vec y), \r'}\;.\ee
The derivatives of the generating functional computed at zero external fields equal the Euclidean correlation functions, cf. with, e.g., \cite[Eq.(27),(28)]{GJMP}. 
Needless to say, the Euclidean correlations satisfy non trivial Ward Identities, following from the lattice continuity equation. For an example, cf. with \cite[Eq.(19),(20)]{GJMP}. 

\medskip

In order to compute the generating functional $\mathcal{W}(f,A)$ in Eq. (\ref{eq.11}), we use an expansion in $U$, 
which is convergent uniformly in the volume and temperature, and uniformly close to (and even on) the critical lines $m_{\text{R}, \pm} = 0$. 
Note that, in the parameter range \eqref{eq.parrange1} the propagator $\hat g(\kk)$ is singular only when $m_{\text{R},-}=0$, in which case the singularity is located at 
$\kk_{F}^{-} := (0, \vec k_{F}^{-})$, with $\vec k_{F}^{\pm} = \big( \frac{2\pi}{3}, \pm \frac{2\pi}{3\sqrt{3}} \big)$. Due to this singularity, the Grassmann integral has, a priori, an infrared problem,
which we resolve by a multi-scale re-summation of the corresponding singularities.

The multi-scale computation of the generating function proceeds as follows. 
First of all, we distinguish the ultraviolet modes, corresponding to large values of the Matsubara frequency, from the infrared ones, by introducing two 
compactly supported cut-off functions, $\chi_\pm(\kk)$, supported in the vicinity of the Fermi points $\kk_F^\pm=(0,\vec k_F^\pm)$;
more precisely, we let $\c_\pm(\kk)=\c_0(\kk- \kk_F^\pm)$, where
$\c_0$ is a smooth characteristic function of the ball of radius $a_0$, with $a_0$ equal to, say, $1/3$)
and by letting $\chi_{\text{uv}}(\kk)=1-\sum_{\o=\pm}\chi_\o(\kk)$.
We correspondingly split the propagator in its ultraviolet and infrared components: 
\be
g(\xx,\yy)=g^{(1)}(\xx,\yy)+\sum_{\o=\pm} e^{-i\vec k_F^\o(\vec x-\vec y)}g_{\o}^{(\le 0)}(\xx,\yy)\label{eq:4}
\ee
where $g^{(1)}(\xx,\yy)$ and $g_{\o}^{(\le 0)}(\xx,\yy)$ are defined in a way similar to Eq.\eqref{prop}, with $\hat g(\kk)$ replaced by $\chi_{\text{uv}}(\kk) \hat g(\kk)$ and by $\chi_{0}(\kk) \hat g(\kk+\kk_F^\o)$,
respectively.  We then split the Grassmann field as a sum of two independent fields, with propagators $g^{(1)}$ and $g^{(\le 0)}$: 
$$\Psi_{\xx,\s}^\pm=\Psi^{\pm(1)}_{\xx,\s}+\sum_{\o=\pm }e^{\pm i\vec k_F^\o\vec x} \Psi_{\xx,\s,\o}^{\pm(\le0)}$$
and we rewrite the Grassmann Gaussian integration as the product of two independent Gaussians: $P(d\Psi)=P(d\Psi^{(\le 0)})P(d\Psi^{(1)})$. By construction, the integration of the `ultraviolet' field $\Psi^{(1)}$ does not have any infrared singularity and, therefore, can be performed in a straightforward manner, thus allowing us to rewrite the generating function $\mathcal W(f,A)$ as the logarithm of 
\be \frac{e^{\WW^{(0)}(f,A)}}{\mathcal N_0}\int P(d\Psi^{(\le 0)})e^{-V^{(0)}(\Psi^{(\le 0)})+B^{(0)}(\Psi^{(\le 0)}, f, A)},\label{eq:v0}
\ee
where $V^{(0)}$ and $B^{(0)}$ are, respectively, the effective potential and the effective source: they are defined by the conditions that $V^{(0)}(0)=0$ and 
$B^{(0)}(0,f,A)=B^{(0)}(\Psi,0,0)=0$. The normalization constant $\mathcal N_0$ is fixed in such a way that $\mathcal N_0=\int P(d\Psi^{(\le 0)})e^{-V^{(0)}(\Psi^{(\le 0)})}$.
All $\mathcal W^{(0)}$, $V^{(0)}$ and $B^{(0)}$ are expressed as
series of monomials in the $\Psi,f,A$ fields, whose kernels (given by the sum of all possible Feynman diagrams with fixed number and fixed space-time location of the external legs)
are {\it analytic functions} of the interaction strength, for $U$ sufficiently small. The precise statement and the proof of these claims are essentially 
identical to those of \cite[Lemma 2]{GM}, see also \cite[Lemma 5.2]{GMPhall} or \cite[Section 6]{Ale}; details will not belabored here and are left to the reader. 

\medskip

In order to integrate the infrared scales, one has to exploit certain lattice symmetries of the model (which replace those of \cite[Lemma 1]{GM}), which allow us to reduce the 
number of independent {\it relevant} and {\it marginal} terms generated by the multi-scale integration. In particular, the symmetries under which the effective potential $V^{(0)}(\Psi)$ is invariant are the following \cite[Sect.III.B]{GJMP}. 

\medskip

\noindent (1) {\it Discrete rotation:}
\be\label{eq.4.6}
\hat \Psi_{\kk',\s,\o}^-\to e^{i\o \frac{2\pi}{3}n_-}e^{-i\vec k'\cdot\vec\ell_2\,n_-}\hat \Psi_{T\kk',\s,\o}^-\;,\quad 
\hat\Psi_{\kk',\s,\o}^+\to\hat\Psi_{T\kk',\s,\o}^+e^{i\vec k'\cdot\vec\ell_2\,n_-}e^{-i\o \frac{2\pi}{3}n_-}
\ee
where, denoting the Pauli matrices by $\s_1,\s_2,\s_3$, we defined
\be
n_-=(1-\s_3)/2\;,\qquad T\kk'=(k_0',e^{-i\frac{2\p}{3}\s_2}\vec k')\;;
\ee
that is, $T$ is the spatial rotation by $2\p/3$ in the counter-clockwise direction.
\medskip

\noindent (2) {\it Complex conjugation}:
\be
\hat\Psi^{\pm}_{\kk',\s,\o}\rightarrow \hat\Psi^{\pm}_{-\kk',\s,-\o},\quad 
c\rightarrow c^{*}\;,\quad \phi\to-\phi\;,
\ee
where $c$ is a generic constant appearing in $P(d\Psi)$ or in $V(\psi)$.

\medskip

\noindent (3) {\it Horizontal reflections}:
\be
\hat\Psi^{-}_{\kk',\s,\o}\to \s_1\hat\Psi^-_{R_h\kk',\s,\o}\;,\quad \hat\Psi^{+}_{\kk',\s,\o}\to \hat\Psi^+_{R_h\kk',\s,\o}\s_1
\,,\quad (W,\phi)\to (-W,-\phi)
\ee
where $R_h\kk'=(k_0',-k_1',k_2')$.

\medskip

\noindent (4) {\it Vertical reflections}:
\be
\hat\Psi^{\pm}_{\kk',\s,\o}\rightarrow
\hat\Psi^{\pm}_{R_v\kk',\s,-\o}\;,\quad \phi\to-\phi.
\ee
where $R_v\kk'=(k_0',k_1',-k_2')$.

\medskip

\noindent (5) {\it Particle-hole}:
\be
\hat\Psi^{-}_{\kk',\s,\o}\to i\hat\Psi^{+,T}_{P\kk',\s,-\o}\;,\quad \hat\Psi^{+}_{\kk',\s,\o}\to i\hat\Psi^{-,T}_{P\kk',\s,-\o}\;,\quad \phi\to-\phi\;.
\ee
where $P\kk'=(k_0',-k_1',-k_2')$.

\medskip

\noindent (6) {\it Magnetic reflections}:
\be
\hat\Psi^{-}_{\kk',\s,\o}\to-i \s_1\s_3\hat\Psi^-_{-R_v\kk',\s,\o}\;,\quad \hat\Psi^{+}_{\kk',\s,\o}\to -i\hat\Psi^+_{-R_v\kk',\s,\o}\s_3\s_1\;,\quad \phi\to\pi-\phi.\label{eq.4.12}
\ee

\medskip

These symmetries have nonperturbative consequences of the structure of the effective interaction action $V^{(0)}$. 
%
At fixed $W,\phi$, the theory is invariant under the transformations (1), (2)+(4), and (2)+(5). In particular, these transformations leave the quadratic part 
\be \label{ciocioe}Q^{(0)}(\Psi)=\sum_{\s,\o}\int \frac{d\kk'}{2\p|\mathcal B|}\,\hat \Psi^+_{\kk',\s,\o}\hat W_{2;\o}^{(0)}(\kk')\hat \Psi^-_{\kk',\s,\o}\ee
of the effective potential $V^{(0)}(\Psi)$ invariant (in \eqref{ciocioe}, $\int \frac{d\kk'}{2\p|\mathcal B|}$ is a shorthand for the Riemann sum $(\b L^2)^{-1}
\sum_{k_0\in \frac{2\pi}\b\mathbb Z}\sum_{\vec k\in\mathcal B_L}$). 
This means that:
\bea\label{symm} \hat W_{2;\o}^{(0)}(\kk')&=&e^{-i(\o\frac{2\pi}{3}+\vec k'\cdot\vec\ell_1)n_-} \hat W_{2;\o}^{(0)}(T^{-1}\kk')e^{i(\o\frac{2\pi}{3}+\vec k'\cdot\vec \ell_1)n_-}\\
&=&\big[\hat W_{2;\o}^{(0)}(-k_0',-k_1',k_2')\big]^*=\big[\hat W_{2;\o}^{(0)}(-k_0',k_1',k_2')\big]^\dagger.\nn\eea
The values of $\hat W_{2;\o}^{(0)}(\kk')$ and of its derivatives at $\kk'=\V0$ define the {\it effective coupling constants}. By computing Eq.\eqref{symm} at $\kk'=\V0$, 
we find, for $\o=\pm$,
\be 
\hat W_{2;\o}^{(0)}(\V0)=e^{-i\frac{2\p}3\o n_-}\hat W_{2;\o}^{(0)}(\V0)e^{i\frac{2\p}3\o n_-}= \big[\hat W_{2;\o}^{(0)}(\V0)\big]^{*}=\big[\hat W_{2;\o}^{(0)}(\V0)\big]^\dagger\;.
\ee
This implies:
\be \hat W_{2;\o}^{(0)}(\V0)=\xi_{\o,0}+\d_{\o,0}\s_3,\label{eq:w2_loc0}\ee
for two {\it real} constants $\xi_{\o,0}$ and $\d_{\o,0}$. Let us now discuss the structure of the derivative of the kernel of the quadratic terms. By taking the derivative of Eq.(\ref{symm}) w.r.t. $\kk'$ and then setting $\kk'=\V0$, we get: 
\be\label{symm_der}
\partial_{\kk'}\hat W_{2;\o}^{(0)}(\V0)=e^{-i\frac{2\p}3\o n_-}T\partial_{\kk'}\hat W_{2;\o}^{(0)}(\V0)e^{i\frac{2\p}3\o n_-}=(-R_v)\partial_{\kk'} \hat W_{2;\o}^{(0)*}(\V0)
=(-P)\partial_{\kk'}\hat W_{2;\o}^{(0)\dagger}(\V0),\ee
where $R_v$ (resp. $P$) is the diagonal matrix with diagonal elements $(1,1,-1)$ (resp. $(1,-1,-1)$). Eq.(\ref{symm_der}) implies that:
\be \kk'\partial_{\kk'}\hat W_{2;\o}^{(0)}(\V0) =
\begin{pmatrix}  -i z_{1,\o} k_{0}' & -u_{\o}(-i k_{1}' +\o k_{2}') \\ -u_{\o}(i k_{1}' + \o k_{2}') & -i z_{2,\o} k_{0}'  \end{pmatrix},\label{eq:6}\ee
where $u_\o, z_{1,\o},z_{2,\o}$ are {\it real} constants.
\medskip

The integration of $\Psi^{(\le 0)}_\o$ is performed iteratively. One rewrites $\Psi^{(\leq 0)}_{\o} = \sum_{h\leq 0} \Psi^{(h)}_{\o}$, for suitable single-scale fields $\Psi^{(h)}_{\o}$. The covariance $\hat g^{(h)}_{\o}$ of $\Psi^{(h)}_{\o}$, supported for quasi-momenta $\kk'$ such that $a_{0}2^{h-1}\leq |\kk'|\leq a_{0}2^{h+1}$, will be defined inductively. We consider two different regimes. The first corresponds to scales $h\ge h^*_1$, with 
\be h^*_1:=\min\{0,\lfloor \log_2 m_{R,+}\rfloor\}, \label{h*1}\ee
and the rest to scales $h^*_1\ge h\ge h^*_2$ with 
$h^*_2:=\min\{0,\lfloor \log_2 |m_{R,-}|\rfloor\}$  (recall that we are focusing on the case that 
$m_{R,+}>|m_{R,-}|$.). We describe the iteration in an inductive way. Assume that the fields $\Psi^{(0)}, \Psi^{(-1)},\ldots,\Psi^{(h+1)}$, $h\ge h^*_1$, have been integrated out
and that after their integration the generating function has the following structure, analogous to the one at scale $0$:
\be e^{ \WW(f,A) }=
\frac{e^{\WW^{(h)}(f,A)}}{\mathcal N_h}\int P(d\Psi^{(\le h)})e^{-V^{(h)}(\Psi^{(\le h)})+B^{(h)}(\Psi^{(\le h)}, f, A)},\label{eq:vh}
\ee
where $V^{(h)}$ and $B^{(h)}$ are, respectively, the effective potential and source terms, satisfying the conditions that $V^{(h)}(0)=0$ and 
$B^{(h)}(0,f,A)=B^{(h)}(\Psi,0,0)=0$. 
The normalization constant $\mathcal N_h$ is fixed in such a way that $\mathcal N_h=\int P(d\Psi^{(\le h)})e^{-V^{(h)}(\Psi^{(\le h)})}$.
Here, $P(d\Psi^{(\le h)})$ is the Grassmann Gaussian integration with propagator (diagonal in the $\s$ and $\o$ indices)
\be g^{(\le h)}_\o(\xx,\yy)=\int P(d\Psi^{(\le h)})\Psi^{-(\le h)}_{\xx,\s,\o}\Psi^{+(\le h)}_{\yy,\s,\o}
=\int\frac{d\kk'}{(2\p)^{3}}\,
e^{-i\kk'(\xx-\yy)} \hat g_\o^{(\le h)}(\kk'),\nn\ee
where, letting 
\bea && r_\o(\vec k')= R(\vec k'+\vec k_F^\o),\quad s_\o(\vec k')=-[\O(\vec k'+\vec k_F^{\,\o})-\frac32(ik_1'+\o k_2')],\\
&& 
m_{-}(\vec k')=m_{\text{R},-}+2t_2\big(\a_2(\vec k'+\vec k_F^-)-\a_2(\vec k_F^-)\big)\sin\phi,\\
&& m_{+}(\vec k')=m_{\text{R},-}+6\sqrt3t_2\sin\phi+2t_2\big(\a_2(\vec k'+\vec k_F^+)-\a_2(\vec k_F^+)\big)\sin\phi,\eea
and $\c_h(\kk')=\c_0(2^{-h}\kk')$, 
\be\hat g_\o^{(\le h)}(\kk')=\c_h(\kk') \begin{pmatrix} a_{1,\o,h}(\kk') & b^*_{\o,h}(\kk')\\
b_{\o,h}(\kk') & a_{2,\o,h}(\kk')\end{pmatrix}^{\!\!\!-1},\label{eq:prop.2}\ee
with
\bea && a_{\r,\o,h}(\kk)=-ik_0Z_{\r,\o,h}+r_\o(\vec k')+(-1)^{\r-1} m_{\o}(\vec k'),\nn\\
&& b_{\o,h}(\kk')=-v_{\o,h} (ik_1'+\o k_2')+s_\o(\vec k')\;,\label{eq:prop.3}
\eea
and the understanding that $(-1)^{\rho-1}$ is equal to $+1$, if $\rho=A$, and equal to $-1$, if $\rho=B$.
The quantities $Z_{\r,\o,h}$ and $v_{\o,h}$ are {\it real}, and they have, respectively, the meaning of wave function renormalizations and of effective velocities.
Note that $r_\o(\vec k')$ and $s_\o(\vec k')$ are both of order $O(|\vec k'|^2)$, while the mass satisfies (again, recall that $m_{\text{R},+}=m_{\text{R},-}+6\sqrt3 t_2\sin\phi$):
$$m_{\o}(\vec k')=m_{\text{R},\o}+t_\o(\vec k'), \quad \text{with}\quad t_\o(\vec k')=O(|\vec k'|^2).$$
By definition, the representation above is valid at the initial step, $h=0$. In order to inductively prove its validity at the generic step, let us discuss how to pass from scale $h$ to scale $h-1$, that is, how to integrate out the field $\Psi^{(h)}$, and how to re-express the resulting effective theory in the form \eqref{eq:vh}, with $h$ replaced by $h-1$. 
Before integrating the $\Psi^{(h)}$ field out, we split $V^{(h)}$ and $B^{(h)}$ into their {\it local} and {\it irrelevant} parts (here, for 
simplicity, we spell out the definitions only in the $f=0$ case): 
$V^{(h)}=\LL V^{(h)}+\RR V^{(h)}$ and $B^{(h)}=\LL B^{(h)}+\RR B^{(h)}$, where, denoting the quadratic part of $V^{(h)}$ by $$Q^{(h)}(\Psi) = \sum_{\o,\s}\int \frac {d\kk'}{(2\p)^{3}}\, \hat \Psi^+_{\kk',\s,\o} \hat W^{(h)}_{2;\o}(\kk')
\hat \Psi^-_{\kk',\s,\o},$$ and
the part of $B^{(h)}$ of order $(2,0,1)$ in $(\psi,f,A)$ by $$Q^{(h)}(\Psi, A) = \sum_{\o,\s,\mu}\int\frac{d\pp}{(2\p)^3}\int \frac {d\kk'}{(2\p)^{3}}\, \hat A_{\pp,\m}\hat \Psi^+_{\kk'+\pp,\s,\o} \hat W^{(h)}_{2,1;\m,\o}(\kk',\pp)\hat \Psi^-_{\kk',\s,\o}$$
we let:
\be \LL V^{(h)}(\Psi)=\sum_{\o,\s}\int \frac {d\kk'}{(2\p)^{3}}\,
\hat \Psi^+_{\kk',\s,\o} [\hat W^{(h)}_{2;\o}({\bf 0})+\kk'\partial_{\kk'}\hat W^{(h)}_{2;\o}({\bf 0})\big]\hat \Psi^-_{\kk',\s,\o},\nn\ee
and 
\be \LL B^{(h)}(\Psi,0,A)=\sum_{\o, \s, \m}\int\frac{d\pp}{(2\p)^3}\int \frac {d\kk'}{(2\p)^{3}}\,
\hat A_{\pp,\m}\hat \Psi^+_{\kk'+\pp,\s,\o} \hat W^{(h)}_{2,1;\m,\o}({\bf 0},{\bf 0})\hat \Psi^-_{\kk',\s,\o}.\nn
\ee
By the symmetries of the model,
\bea  \LL V^{(h)}(\Psi)&=&\sum_{\o, \s}\int \frac {d\kk'}{(2\p)^{3}}\, \Big[2^h\xi_{\o,h}\hat \Psi^+_{\kk',\s,\o} \hat \Psi^-_{\kk',\s,\o}+
2^h\d_{\o,h}\hat \Psi^+_{\kk',\s,\o} \s_3 \hat \Psi^-_{\kk',\s,\o}
\label{eq:lv}\\
&&+\hat \Psi^+_{\kk',\s,\o} \begin{pmatrix} -i z_{1,\o,h} k_{0}& -u_{\o,h}(-i k_{1}' +\o k_{2}') \\ -u_{\o,h}(i k_{1}' + \o k_{2}') & -i z_{2,\o,h} k_{0} \end{pmatrix}
\hat \Psi^-_{\kk',\s,\o}\Big],\nn\eea
where $\xi_{\o,h}, \d_{\o,h},z_{\r,\o,h}, u_{\o,h}$ are real constants and $\s_3$ is the third Pauli matrix. We also denote by $\g_{\m,\o,h}:=\hat W^{(h)}_{2,1;\m,\o}({\bf 0},{\bf 0})$ the 
{\it vertex functions}, entering the definition of $\LL B^{(h)}(\Psi,0,A)$. Notice that their structure is constrained by the Ward Identities. E.g., using \cite[Eq.(20)]{GJMP}, one finds 
that $\g_{0,\o,h}=-\sum_{\r=1}^2(Z_{\r,\o,h}+z_{\r,\o,h})n_\r$ (where $n_\r=(1+(-1)^{\r-1}\s_3)/2$), $\g_{1,\o,h}=-(v_{\o,h}+u_{\o,h})\s_2$, and $\g_{2,\o,h}=-\o(v_{\o,h}+u_{\o,h})\s_1$. However, in the following, we will neither need these identities, nor to identify any special structure of $\g_{\mu,\o,h}$, with $\mu=3,4$. 

\medskip

Once the effective potential and source have been split into local and irrelevant parts, we combine the part of $\LL V^{(h)}$ in the second line of \eqref{eq:lv} with the 
Gaussian integration $P(d\Psi^{(\le h)})$, thus defining a dressed measure $\tilde P(d\Psi^{(\le h)})$ whose propagator $\tilde g^{(\le h)}_\o(\xx,\yy)$ is 
analogous to $g^{(\le h)}_\o(\xx,\yy)$, with the only difference that the functions $a_{\r,\o,h}$, $b_{\o,h}$ in \eqref{eq:prop.2}-\eqref{eq:prop.3}
are replaced by \bea \tilde a_{\r,\o,h-1}(\kk)&=&-ik_0\tilde Z_{\r,\o,h-1}(\kk')+r_\o(\vec k')+(-1)^{\r-1} m_{\o}(\vec k'),\nn\\
\tilde b_{\o,h-1}(\kk')&=&-\tilde v_{\o,h-1}(\kk') (ik_1'+\o k_2')+s_\o(\vec k'),\nn\eea 
with \bea &&\tilde Z_{\r,\o,h-1}(\kk')=Z_{\r,\o,h}+z_{\r,\o,h}\,\c_h(\kk'),\nn\\
&& \tilde v_{\o,h-1}(\kk')=v_{\o,h}+u_{\o,h}\,\c_h(\kk').\nn\eea
Now, by rewriting the support function $\c_h(\kk')$ in the definition of $\tilde g^{(\le h)}_\o(\xx,\yy)$ as $\c_h(\kk')=f_h(\kk')+\c_{h-1}(\kk')$,
we correspondingly rewrite: $\tilde g^{(\le h)}_\o(\xx,\yy)=\tilde g^{(h)}_\o(\xx,\yy)+g^{(\le h-1)}_\o(\xx,\yy)$, where $g^{(\le h-1)}_\o(\xx,\yy)$ is defined exactly as in 
\eqref{eq:prop.2}-\eqref{eq:prop.3}, with $h$ replaced by $h-1$, and $Z_{\r,\o,h-1}, v_{\o,h-1}$ defined by the flow equations:
\be Z_{\r,\o,h-1}=Z_{\r,\o,h}+z_{\r,\o,h},\qquad  v_{\o,h-1}=v_{\o,h}+u_{\o,h}.\nn\ee
We are now ready to integrate the fields on scale $h$. We define:
\bea &&e^{-V^{(h-1)}(\Psi)+B^{(h-1)}(\Psi,f,A)+w^{(h)}(f,A)}=C_h\int \tilde P(d\Psi^{(h)}) e^{-F_\xi^{(h)}(\Psi^{(h)}+\Psi)-F_\d^{(h)}(\Psi^{(h)}+\Psi)}\times\nn\\
&&\hskip3.truecm \times 
e^{-\RR V^{(h)}(\Psi^{(h)}+\Psi)+\mathcal L B^{(h)}(\Psi^{(h)}+\Psi, f, A)+
\mathcal R B^{(h)}(\Psi^{(h)}+\Psi, f, A)},\eea
where $\tilde P(d\Psi^{(h)})$ is the Gaussian integration with propagator $\tilde g^{(h)}_\o$, 
$$F_\xi^{(h)}(\Psi)=\sum_{\o}2^h\xi_{\o,h}\int \frac {d\kk'}{(2\p)^{3}}\hat \Psi^+_{\kk',\s,\o} \hat \Psi^-_{\kk',\s,\o},
\qquad F_\d^{(h)}(\Psi)=\sum_{\o}2^h\d_{\o,h}\int \frac {d\kk'}{(2\p)^{3}}\hat \Psi^+_{\kk',\s,\o} \s_3\hat \Psi^-_{\kk',\s,\o},$$ and $C_h^{-1}=
\int \tilde P(d\Psi^{(h)})e^{-F_\xi^{(h)}(\Psi^{(h)})+\RR V^{(h)}(\Psi^{(h)})}$. Finally, letting $\mathcal W^{(h-1)}=\mathcal W^{(h)}+w^{(h)}$, we obtain the 
same expression as \eqref{eq:vh}, with $h$ replaced by $h-1$. This concludes the proof of the inductive step, corresponding to the integration of the fields on scale $h$, with 
$h\ge h^*_1$. By construction, the running coupling constants $\vec\t_h=(\x_{\o,h},\d_{\o,h}, Z_{A,\o,h}, Z_{B,\o,h},v_{\o,h})_{\o\in\{\pm\}}$
verify the following recursive equations:
\bea
\xi_{\o,h-1}&=&2\xi_{\o,h}+\b^\xi_{\o,h}(U, \vec \t_h,\ldots,\vec \t_0),\nn\\
\d_{\o,h-1}&=&2\d_{\o,h}+\b^\d_{\o,h}(U,\vec \t_h,\ldots,\vec \t_0),\label{eq:flow}\\
Z_{\r,\o,h-1}&=&Z_{\r,\o,h}+\b^{Z,\r}_{\o,h}(U,\vec \t_h,\ldots,\vec \t_0),\nn\\
v_{\o,h-1}&=&v_{\o,h}+\b^v_{\o,h}(U,\vec \t_h,\ldots,\vec \t_0),\nn\eea
for suitable functions $\b^\sharp_{\o,h}$, known as the (components of the) {\it beta function}. Note that the initial data $\xi_{\o,0},\d_{\o,0},Z_{\r,\o,0},v_{\o,0}$ are 
analytically close to $\xi,\d,1,\frac32 $, respectively; they are not exactly independent of the indices $\r,\o$, due to to the effect of the ultraviolet integration. However, for small 
values of $m_{\text{R},+}$, the difference between the initial data, for different values of the indices, differ at most by $O(U m_{\text{R},+})$ (note that $m_{\text{R},+}=O(|m_{\text{R},+}|+\sin\phi)$).
As we shall see below, the running coupling constants  remain analytically close to their initial data, for all $h\le 0$.
Similarly, the vertex functions satisfy recursive equations driven by the running coupling constants themselves:
$$\g_{\m,\o,h-1}=\g_{\m,\o,h}+\sum_{h'=h}^0\g_{\m,\o,h'}\,\tilde\b^\g_{\m,\o,h'}(U,\vec \t_h,\ldots,\vec \t_0)\;,$$
whose solution remains analytically close to the corresponding initial data, for all $h\le 0$. 

From the structure and properties of the effective propagator on scale $h$, see \eqref{eq:prop.2} and following lines, one recognizes that the effective theory at scale $h$ 
is a lattice regularization of a theory of relativistic fermions with masses $m_{\text{R}, \pm}$. As anticipated above, $Z_{\r,\o,h}$ and $v_{\o,h}$ remain analytically 
close to their initial data $1,\frac32$, for all $h\le 0$: therefore, it is straightforward to check that the single scale propagator satisfies
\be
|g^{(h)}_\o(\xx,\yy)|\le C_N \frac{2^{2 h}}{1+(2^h|\xx-\yy|)^N }\;,\qquad \forall N\ge 1\;.
\ee
Moreover, the single-scale propagator admits the decomposition:
\be
g^{(h)}_\o(\xx,\yy) = G^{(h)}_{\o}(\xx,\yy) + g^{(h)}_{\o,r}(\xx,\yy)\label{ao}
\ee
where $G^{(h)}_{\o}(\xx,\yy)$ is obtained from $g^{(h)}_\o(\xx,\yy)$ by setting $m_{\text{R},\o}=0$, and where the remainder term $g^{(h)}_{\o,r}$ satisfies the same bound as $g^{(h)}_\o$ times an extra factor $m_{\text{R}, \o} 2^{-h}$, which is small, for all scales larger than $h^*_1$.

Due to the fact that $m_{\text{R}, +} \geq |m_{R, -}|$, once we reach the scale $h = h^{*}_{1}$, the infrared propagator of the field corresponding to $\o = +$ satisfies the following bound:
\be\label{eq:h1+}
|g^{(\le h_1^*)}_+(\xx,\yy)|\le C_N \frac{2^{2 h^*_1}}{1+(2^{ h^*_1}|\xx-\yy|)^N }\;;
\ee
that is, it admits the same qualitative bound as the corresponding single scale propagator on scale $h = h^{*}_{1}$. For this reason, it can be integrated in a single step, without any 
further need for a multiscale analysis. We do so and, after its integration, we are left with an effective theory on scales $h\le h^*_1$, depending only on $\Psi^{(\leq h^{*}_{1})}_{-}$, 
which we integrate in a multiscale fashion, similar to the one described above, until the scale $h = h^{*}_{2}$ is reached. At that point, the infrared 
propagator $g^{(\le h_2^*)}_-$ satisfies a 
bound similar to (\ref{eq:h1+}), with $h^{*}_{1}$ replaced by $h^{*}_{2}$, and the corresponding field can be integrated in a single step. The outcome of the final integration is the 
desired generating function. 

\medskip

The iterative integration procedure described above provides an explicit algorithm for computing the kernels of the effective potential and sources. In particular, 
they can be represented as sums of {\it Gallavotti-Nicol\` o trees}, identical to those of \cite[Section 3]{GM}, modulo the following minor differences. 
The endpoints $v$ on scale $h_{v} = +1$ are associated either with $F^{(0)}_\x(\Psi^{(\le 0)})$, or with $F^{(0)}_\d(\Psi^{(\le 0)})$, or with $\LL B^{(0)}(\Psi^{(\le 0)})$, or with one of the 
terms in $\RR\VV^{(0)}(\Psi^{(\le 0)})$ or in $\RR B^{(0)}(\Psi^{(\le 0)},f,A)$; the endpoints on scale $h_{v}\le 0$ are, instead, associated either with $F^{(h_v-1)}_\x(\Psi^{(\le h_v-1)})$, 
or with $F^{(h_v-1)}_\d(\Psi^{(\le h_v-1)})$, or with $\LL B^{(h_v-1)}(\Psi^{(\le h_v-1)},f,A)$. The most important novelty of the present construction, as compared with 
\cite{GM}, is the presence of the relevant couplings $\x_{\o,h},\d_{\o,h}$, whose flow must be controlled by properly choosing the counterterms $\xi$ and $\d$, see discussion below. 
Recall that the flows of $\xi_{+,h}$ and $\d_{+,h}$ stop at scale $h^*_1$; for smaller scales, we let $\x_{+,h'}=\d_{+,h'}=0$, $\forall h'<h^*_1$.
Similarly, we let the other running coupling constants with $\o=+$, that is, $Z_{\r,+,h}$ and $v_{+,h}$, be zero for scales smaller than $h^*_1$.
It turns out that the tree expansion is {\it absolutely convergent}, provided that $U$ is small enough and the relevant couplings remain small, uniformly in the scale $h\le 0$.  
More precisely, the kernels of the effective potential satisfy the following bound (a similar statement is valid, of course, for the kernels of the effective source). Notation-wise, 
we let $W_n^{(h)}(\xx_1,\ldots,\xx_n)$ be the kernel of the effective potential $\mathcal V^{(h)}(\Psi)$ associated with the monomial in $\Psi$ of order $n$; of course, $W^{(h)}_n$ is non 
zero only if $n$ is even. The arguments $\xx_1,\ldots,\xx_n$ are the space-time coordinates of the Grassmann fields; the kernel implicitly depends also on the $\rho,\o$ indices of the 
external fields, but we do not spell out their dependence explicitly. We also let $\|W^{(h)}_n\|_1:=\int d\xx_2\cdots d\xx_n |W^{(h)}_n(\xx_1,\ldots,\xx_n)|$ (here $\int d\xx$ is a shorthand 
for $\int_0^\b dx_0\sum_{\vec x\in\L_L}$), which is  independent of $\xx_1$, due to translational invariance. 

\begin{lemma}\label{lem:W} There exist positive constants $U_0, \th, C_0$, such that the following is true. 
Suppose that $\max_{\r,\o,k\ge h}\{|Z_{\rho,\o,k} - 1|,|v_{\o, k} - \frac32|,|\x_{\o, k}|,|\d_{\o, k}|\}\leq C|U|$. 
Then, the kernels of the effective potential on scale $h-1$ are analytic in $U$ for $|U|\le U_0/(C+1)$, and satisfy the bound
\bea\label{bbasic}
&&\| {W}^{(h-1)}_{2} \|_{1}\leq C|U|2^{h}+C_0|U|2^{h(1+\th)}\;,\\
&& \| {W}^{(h-1)}_{n} \|_{1}\leq C_0^n |U|^{\frac{n}2-1} 2^{h(3 - n+\th)}\,, \qquad \forall n\ge 4\;.\label{bbasic2}
\eea
%
%
The components of the beta function are analytic in $U$ in the same domain, and satisfy:
\be \big|\b^\#_{\o,h}(U, \vec \t_h,\ldots, \vec \t_0)\big|\le C_0 |U| 2^{\th h}\,.\label{eq.bound}\ee
\end{lemma}

%

The proof of the lemma goes along the same lines as the proof of \cite[Theorem 2]{GM}, see also the review \cite{Gi}, and will not be repeated here. Two key ingredients in the proof 
are: the representation of the iterated truncated expectations in terms of the Brydges-Battle-Federbush determinant formula, and the Gram-Hadamard bound. The factors $2^{\th h}$ appearing in the right sides of \eqref{bbasic},
\eqref{bbasic2} and \eqref{eq.bound}, represent a `dimensional gain', as compared to a more basic, naive, dimensional bound, proportional to $2^{(3-n)h}$, which is suggested by the fact that the scaling dimension
of the contributions to the effective potential with $n$ external fermionic is equal to $3-n$, in the RG jargon (we use the convention that positive/negative scaling dimensions correspond to relevant/irrelevant operators). 
Such a dimensional gain is due to the {\it RG irrelevance} of the quartic interaction (note that $3-n=-1$ for $n=4$) and to the so-called short-memory property of the Gallavotti-Nicol\`o trees (``long trees are exponentially suppressed''): all the contributions to the effective potential associated with trees that have at least one endpoint on scale $+1$ have this additional exponentially decaying factor. The only contributions not having such a gain are 
those associated with trees without endpoints on scale $+1$. The key remark is that, since the running coupling constants are all associated with quadratic contributions in the fermionic fields, 
such contributions are very simple and explicit: they can all be represented as sums of linear Feynman diagrams with two external legs (`chain diagrams'), obtained by contracting in all possible ways the two-legged vertices corresponding to the running coupling constants $\xi_{\o,k},\d_{\o,k}$. Therefore, they only contribute to the quadratic part of the effective potential, and they lead to the first term in the right side of \eqref{bbasic}. 
Note also that such diagrams do not contribute to the beta function: in fact, the beta function at scale $h$ is obtained by taking the `local part' of $W_{2}^{(h)}$, which is equal to the value of the Fourier transform 
$\widehat W_2^{(2)}$ at $\kk'=\V0$. If we compute the chain diagrams at $\kk'=\V0$, we see that the quasi-momenta of all the propagators of the chain diagram are equal to zero;  therefore, the value of the diagram is zero, too, 
due to the compact support properties of the single-scale propagator. 

\medskip

The idea, now, is to use the bound on the beta function to inductively prove the assumption on the running coupling constants, or, more precisely, the following improved version
of the inductive assumption:
\bea  & |Z_{\rho,\o,h} - 1|\leq C|U|,\quad |v_{\o, h} - \frac32|\leq C|U|,& \forall h^*_2\le h\le 0\;,\nn\\
&|\x_{-, h}|\leq C|U|2^{\th h},\quad |\d_{-, h}|\leq C|U|2^{\th h},& \forall h^*_2\le h\le 0\;,
\label{eq.ind}\\
& |\x_{+, h}-\x_{-,h}|\leq C|U|2^{h^*_1-h},\quad |\d_{+, h}-\d_{-,h}|\leq C|U|2^{h^*_1-h},& \forall h^*_1\le h\le 0\;,\nn\eea
for a suitable $C>0$ (recall that, by definition, $\xi_{+,h}=\d_{+,h}=Z_{\r,+,h}=v_{+,h}=0$, $\forall h<h^*_1$). 
Note that the bound on the beta function is already enough to prove the assumption for 
$Z_{\rho,\o,h}$ and $v_{\o,h}$. The subtle point is to control the flow of $\x_{\o, h}$, $\d_{\o, h}$, 
provided the initial data $\x,\d$ are properly chosen. This is the content of the next lemma. 

\begin{lemma}\label{lem:count} There exist positive constants $U_0$, $C$, and functions $\d=\d(U,m_{\text{R},-},\phi)$, $\xi=\xi(U,m_{\text{R},-},\phi)$, analytic in $U$ for $|U|\le U_0/(C+1)$ and vanishing at $U=0$, such that \eqref{eq.ind} 
are verified.
\end{lemma}

\begin{proof} We solve the beta function by looking at it as a fixed point equation on a suitable space of sequences.
The fixed point equation arises by iterating the beta function equation and then imposing that $\x_{-,h^*_2}=\d_{-,h^*_2}=0$. 
By iterating the first two equations of (\ref{eq:flow}), we get, for all $h^*_2\le h\le 0$,
\bea\label{eq:map}
\xi_{\o,h} &=& 2^{-h}\big(\xi_{\o,0}+\sum_{k=h+1}^{0} 2^{k-1} \b^\xi_{\o,k}(U, \vec \t_k,\ldots,\vec \t_0)\big)\\
\d_{\o,h} &=& 2^{-h}\big(\d_{\o,0}+\sum_{k=h+1}^{0} 2^{k-1}\b^\d_{\o,k}(U,\vec \t_h,\ldots,\vec \t_0)\big)\;,\nn
\eea
with the understanding that $\xi_{+,h}=\d_{+,h}=0$, $\forall h<h^*_1$.  
Consider first the case $\o=-$. By imposing the condition that $\x_{-,h^*_2}=\d_{-,h^*_2}=0$, we find that 
\be \label{in.dat}\xi_{-,0}=-\sum_{k=h+1}^{0} 2^{k-1} \b^\xi_{-,k}(U,\vec \t_h,\ldots,\vec \t_0),\qquad \d_{-,0}=-\sum_{k=h+1}^{0} 2^{k-1}\b^\d_{-,k}(U,\vec \t_h,\ldots,\vec \t_0).\ee
Plugging these identities back in \eqref{eq:map} with $\o=-$ gives 
\be \label{xd0}\xi_{-,h}=-\sum_{h^*_2<k\le h} 2^{k-h-1}\b^\xi_{-,k}(U,\vec \t_h,\ldots,\vec \t_0),\qquad \d_{-,h}=-\sum_{h^*_2<k\le h} 2^{k-h-1}\b^\d_{-,k}(U,\vec \t_h,\ldots,\vec \t_0),\ee
which is the desired equation for $\xi_{-,h},\d_{-,h}$. Consider next the case $\o =+$. The initial data $\xi_{+,0}, \d_{+,0}$ in the right side of \eqref{eq:map} are regarded as given functions of $U,\xi_{-,0},\d_{-,0},m_{\text{R},-},\phi$, whose explicit form follows from the ultraviolet integration, such that both $\xi_{+,0}-\xi_{-,0}$ 
and $\d_{+,0}-\d_{-,0}$ are of the order $O(U \min\{m_{\text{R},+},1\})$. More explicitly, we write, 
\be \label{xd+}\xi_{+,0}=\xi_{-,0}+\bar x_+(U,\xi_{-,0},\d_{-,0},m_{\text{R},-},\phi), \qquad \d_{+,0}=\d_{-,0}+\bar d_+(U,\xi_{-,0},\d_{-,0},m_{\text{R},-},\phi),\ee
where $\bar x_+$ and $\bar d_+$ are analytic in $U,\xi_{-,0},\d_{-,0}$ for $|\xi_{-,0}|,|\d_{-,0}|\le C|U|$ and $|U|\le U_0/(C+1)$, and satisfy: 
\bea && |\bar x_+(U,\xi_{-,0},\d_{-,0},m_{\text{R},-},\phi)|\le C_1|U|\min\{m_{\text{R},+},1\},\nn \\
&& |\bar x_+(U,\xi_{-,0},\d_{-,0},m_{\text{R},-},\phi)-\bar x_+(U,\xi_{-,0}',\d_{-,0}',m_{\text{R},-},\phi)|\le \label{good.bound}\\
&& \hskip3.truecm \le C_1|U|\min\{m_{\text{R},+},1\}(|\xi_{-,0}-\xi_{-,0}'|+|\d_{-,0}-\d_{-,0}'|)\;,\nn
\eea
for some $C_1>0$, and analogously for $\bar d_+$. Plugging \eqref{xd+}, with $\xi_{-,0},\d_{-,0}$ written as in \eqref{in.dat}, back in 
\eqref{eq:map} with $\o=+$, we get the desired equation for $\xi_{+,h},\d_{+,h}$:
\bea\label{eq:map.2}
\xi_{+,h} &=& 2^{-h}\big(\bar x_++\sum_{k=h+1}^{0} 2^{k-1} (\b^\xi_{+,k}-\b^\xi_{-,k}) -\sum_{k=h^*_2+1}^h2^{k-1}\b^\xi_{-,k}\big)\;,\\
\d_{+,h} &=& 2^{-h}\big(\bar d_++\sum_{k=h+1}^{0} 2^{k-1} (\b^\d_{+,k}-\b^\d_{-,k}) -\sum_{k=h^*_2+1}^h2^{k-1}\b^\d_{-,k}\big)\;,\nn
\eea
for all $h^*_1\le h\le 0$. The equations \eqref{xd0} and \eqref{eq:map.2}, together with the analogues of \eqref{eq:map} for the running coupling constants $Z_{\r,\o,h},v_{\o,h}$, 
are looked at as a fixed point equation on the space $\mathcal M$ of sequences of 
running coupling constants $\underline{\t}:=\{\vec \t_{h^*_2},\ldots,\vec\t_0\}$, endowed with the norm 
\bea \|\underline\t\|_\th=\max&\Big\{&\max_{\substack{h\le 0\\ \o,\r}}\{|Z_{\r,\o,h}-1|,|v_{\o,h}-\frac32 |,2^{-\th h}|\xi_{-,h}|,2^{-\th h}|\d_{-,h}|\},\nn\\
&\phantom{\Big\{}& \hskip-.2truecm\max_{h^*_1\le h\le 0}\{|\xi_{+,h}-\xi_{-,h}|2^{h-h^*_1},|\d_{+,h}-\d_{-,h}|2^{h-h^*_1}\}\Big\}.\eea
More precisely, the sequence of running coupling constants, solution of the flow equation with initial data such that $\x_{-,h^*_2}=\d_{h^*_2}=0$, 
is the fixed point of the map $\underline \t\to \underline\t'={\bf T}(\underline\t)$ that, in components, reads (we write 
the argument of the beta function as $(U,\underline\t)$, and we do not indicate the argument of $\bar x_+$ and $\bar d_+$, for short):
\bea  
\xi_{-,h}'&=&-\sum_{k=h^*_2+1}^h 2^{k-h-1}\b^\xi_{-,k}(U, \underline\t),\hskip1.truecm \forall h^*_2\le h\le 0\label{xi-h}\\
 \d_{-,h}'&=&-\sum_{k=h^*_2+1}^h 2^{k-h-1}\b^\d_{-,k}(U, \underline\t),\hskip1.truecm\forall h^*_2\le h\le 0\\
 \xi_{+,h}'&=&2^{-h}\bar x_+
+\sum_{k=h+1}^{0} 2^{k-h-1} (\b^\xi_{+,k}(U, \underline\t)-\b^\xi_{-,k}(U, \underline\t))\nn \\
&&- \sum_{k=h^*_2+1}^h2^{k-h-1}\b^\xi_{-,k}(U, \underline\t),\hskip1.truecm\forall h^*_1\le h\le 0\label{xi.1}\\
\d_{+,h}'& =& 2^{-h}\bar d_+
+\sum_{k=h+1}^{0} 2^{k-h-1} (\b^\d_{+,k}(U, \underline\t)-\b^\d_{-,k}(U, \underline\t))\nn \\
&&-\sum_{k=h^*_2+1}^h2^{k-h-1}\b^\d_{-,k}(U, \underline\t),
\hskip1.truecm\forall h^*_1\le h\le 0\label{d.1}\\
Z_{\r,\o,h}'&=&1+\bar z_{\r,\o}
+\sum_{k=h+1}^{0}\b^{Z,\r}_{\o,k}(U, \underline\t),\hskip.97truecm\forall h^*_2\le h\le 0\\
v_{\o,h}'&=&\frac32t_1+\bar v_{\o}
+\sum_{k=h+1}^{0} \b^v_{\o,k}(U, \underline\t)\;,\hskip.73truecm\forall h^*_2\le h\le 0\eea
with the understanding that the running coupling constants with $\o=+$ are zero for all scales smaller than $h^*_1$: $\x_{+,h}=\d_{+,h}=Z_{\r,+,h}=v_{+,h}=0$, for all $h<h^*_1$.
Moreover, in the last two lines, we rewrote $Z_{\r,\o,0}=1+\bar z_{\r,\o}$ and $v_{\o,0}=\frac32+\bar v_\o$, 
where $\bar z_{\r,\o}=\bar z_{\r,\o}(U,\xi_{-,0},\d_{-,0},m_{\text{R},-},\phi)$ and $\bar v_{\o}=\bar v_\o(U,\xi_{-,0},\d_{-,0},m_{\text{R},-},\phi)$ are analytic in $U,\xi_{-,0},\d_{-,0}$ for $|\xi_{-,0}|,|\d_{-,0}|\le C|U|$ and $|U|\le U_0/(C+1)$, and satisfy: 
\bea && |\bar z_{\r,\o}(U,\xi_{-,0},\d_{-,0},m_{\text{R},-},\phi)|\le C_1|U|,\nn \\
&& |\bar z_{\r,\o}(U,\xi_{-,0},\d_{-,0},m_{\text{R},-},\phi)-\bar z_{\r,\o}(U,\xi_{-,0}',\d_{-,0}',m_{\text{R},-},\phi)|\le \label{ggo.bound.2} \\
&&\hskip4.6truecm \le C_1|U|(|\xi_{-,0}-\xi_{-,0}'|+|\d_{-,0}-\d_{-,0}'|)\;,\nn
\eea
and analogously for $\bar v_\o$. In addition, the differences $\bar z_{\r,+}-\bar z_{\r,-}$ and $\bar v_+-\bar v_-$ satisfy the same bound as \eqref{good.bound}. 

We want to show that the map $\underline\t \to {\bf T}(\underline\t)$ admits a unique fixed point in the ball 
$B_0=\{\underline \t\in\mathcal M: \|\underline\t\|_\th\le C|U|\}$, for a suitable $C>0$. In order to prove this, we show that, if $\underline\t,\underline\t'\in B_0$,  
\be\label{dd}\|{\bf T}(\underline\t)\|_\th\le C|U|,\qquad \|{\bf T}(\underline\t)-{\bf T}(\underline\t')\|_\th\le C|U|\, \|\underline\t-\underline\t'\|_\th\,,\ee
for a suitable $C$. Once \eqref{dd} is proved, the existence of a unique fixed point in $B_0$ follows via the Banach fixed point theorem, and we are done:
such a fixed point defines the initial data $\x_{-,0},\d_{-,0}$ generating a solution to the flow equation satisfying \eqref{eq.ind}, as desired. Of course, 
fixing $\x_{-,0},\d_{-,0}$ is equivalent (thanks to the analytic implicit function theorem) to fixing $\x,\d$: therefore, the existence of such a fixed point proves the statement of the lemma. 

Therefore, we are left with proving \eqref{dd}. If $\underline \t\in B_0$, by using the bound \eqref{eq.bound} on the beta function, as well as the assumptions \eqref{good.bound},
\eqref{ggo.bound.2} on the initial data (together with their analogues for $\bar d_+,\bar v_\o$), it is immediate to check that 
\be |Z_{\r,\o,h}'-1|\le C|U|, \qquad |v_{\o,h}'-\frac32|\le C|U|,\qquad |\xi_{-,h}'|\le C|U|2^{\th h},\qquad |\d_{-,h}'|\le C|U|2^{\th h},\ee
for all $h^*_2\le h\le 0$ and a suitable constant $C$. Therefore, in order to check that $\|{\bf T}(\underline\t)\|_\th\le C|U|$, we are left with proving that $\max\{ |\xi_{+,h}'-\xi_{-,h}'|,|\d_{+,h}'-\d_{-,h}'|\}\le C|U|2^{h^*_1-h}$, for all $h^*_1\le h\le 0$. We spell out the argument for $\xi_{+,h}'-\xi_{-,h}'$, the proof for $\d_{+,h}'-\d_{-,h}'$ being exactly the same. By using \eqref{xi-h}-\eqref{xi.1}, we have: $\xi_{+,h}'-\xi_{-,h}'=2^{-h}\bar x_++\sum_{k=h+1}^{0} 2^{k-h-1} (\b^\xi_{+,k}(U, \underline\t)-\b^\xi_{-,k}(U, \underline\t))$. Now, the first term in the right side is bounded by  
$2^{-h}|\bar x_+|\le 2C_1|U|$, for all $h\ge h^*_1$, by \eqref{good.bound} and the very definition of $h^*_1$, \eqref{h*1}. 
In order to bound the sum 
$\sum_{k=h+1}^{0} 2^{k-h-1} (\b^\xi_{+,k}(U, \underline\t)-\b^\xi_{-,k}(U, \underline\t))$, we note that $\b^\xi_{+,k}-\b^\xi_{-,k}$ can be expressed as 
a sum over trees with root on scale $k$, at least an endpoint on scale $+1$ (recall the discussion after the statement of Lemma \eqref{lem:W}) and: either an endpoint corresponding to a difference 
$\xi_{+,k'}-\xi_{-,k'}$, or an endpoint corresponding to $\d_{+,k'}-\d_{-,k'}$, or a propagator $g^{(k')}_{+}-g^{(k')}_-$, with $k'\ge k$. 
The propagator $g^{(k')}_{+}-g^{(k')}_-$ admits a dimensional bound that is the same as $g^{(k')}_{\o}$ times
a gain factor $2^{h^*_1-k'}$; the differences $\xi_{+,k'}-\xi_{-,k'}$ and $\d_{+,k'}-\d_{-,k'}$ are proportional to the same gain factor, due to the assumption that 
$\underline\t\in B_0$. All in all, recalling the basic bound on the beta function, \eqref{eq.bound}, we find a similar bound, improved by the gain factor $2^{h^*_1-k}$: 
$$\big|\b^\xi_{+,k}(U,\underline\t)-\b^\xi_{-,k}(U,\underline\t)\big|\le 2C_0|U| 2^{h^*_1-k}2^{\th k}.$$
This, together with the bound on $2^{-h}\bar x_+$, implies the desired bound, $|\xi_{+,h}'-\xi_{-,h}'|\le C|U|2^{h^*_1-h}$, for all $h^*_1\le h\le 0$ and $C$ sufficiently large. Exactly the 
same argument implies the desired bound for $\d_{+,h}'-\d_{-,h}'$. 

\bigskip

The proof of the second of \eqref{dd} goes along the same lines, and we only sketch it here. A similar argument, discussed in all details, can be found in \cite[Section 4]{BM1}. Let us focus, for simplicity, on the first component of ${\bf T}(\underline\t)-{\bf T}(\underline\t')$, which reads:
$$-\sum_{k=h^*_2+1}^h 2^{k-h-1}\big(\b^\xi_{-,k}(U, \underline\t)-\b^\xi_{-,k}(U, \underline\t')\big).$$
The difference
$\b^\xi_{-,k}(U,\underline\t)-\b^\xi_{-,k}(U,\underline\t')$ can be represented as a sum over trees  with root on scale $k$, at least an endpoint on scale $+1$, and: 
either an endpoint corresponding to a difference $\xi_{\o,k'}-\xi_{\o,k'}'$, or an endpoint corresponding to $\d_{\o,k'}-\d_{\o,k'}'$, 
or a propagator corresponding to the difference between $g^{(k')}_{\o}$ computed at the values $(Z_{\rho,\o,k'},v_{\o,k'})$ of the effective parameters 
and the same propagator computed at $(Z_{\rho,\o,k'}',v_{\o,k'}')$, for some $k'\ge k$. The difference between the propagators computed at different values of 
the effective parameters can be bounded dimensionally in the same way as $g^{(k')}_{\o}$, times an additional factor $\max_{\r,\o}\{|Z_{\rho,\o,k'}-Z_{\rho,\o,k'}'|,
|v_{\o,k'}-v_{\o,k'}'|\}$. Therefore, recalling the basic bound on the beta function, \eqref{eq.bound}, we find a similar bound, multiplied by the norm of the difference between the running coupling constants:
\be \Big|\b^\xi_{-,k}(U,\underline\t)-\b^\xi_{-,k}(U,\underline\t')\Big|\le 2C_0|U| 2^{\theta k} \|\underline\t-\underline\t'\|_\th,\ee
which implies the desired estimate on the first component of ${\bf T}(\underline\t)-{\bf T}(\underline\t')$. A similar argument is valid for the other components, but we will not belabor the details here. \end{proof}

\bigskip

We now have all the ingredients to prove Proposition \ref{assm}. In fact, in view of Lemma \ref{lem:W} and Lemma \ref{lem:count}, we can fix the counterterms 
$\xi,\d$ in such a way that the kernels of the effective potential on all scales are analytic in $U$, uniformly in the scale, and satisfy \eqref{bbasic}. 
A simple by-product of the proof shows that the kernel $W^{(h)}_{n}(\xx_1,\ldots,\xx_n)$ decays faster than any power in the tree distance among the 
space-time points $\xx_1,\ldots,\xx_n$, with a decay length proportional to $2^{-h}$. Analogous claims are valid for the kernels of the effective source term 
and of the generating function. In particular, recalling that the scale $h$ is always larger or equal than $h^*_2$, we have that the kernels of the effective potential, which are nothing else but the multi-point correlation functions, are analytic in $U$ and decay faster than any power in the tree distance among their arguments, with a typical decay length of the order $2^{h^*_2}\sim |m_{\text{R},-}|$. Therefore, for any $m_{\text{R},-}\neq0$, the Fourier transform of any 
multi-point correlation of local operators is $C^\infty$ in the momenta. In the massless case, the correlations are dimensionally bounded like in the graphene case \cite{GM,GMP1}: in particular, the two-point density-density, or current-current correlations decay like $|\xx-\yy|^{-4}$ at large Euclidean space-time separation. For further details about the construction and estimate of the correlation functions, 
the reader is referred to, e.g., \cite{GeM,GMP1}. This concludes the proof of Proposition \ref{assm}. \end{proof}

\section{Proof of Theorem \ref{thm:1}}\label{sec.5}

In order to conclude the proof of Theorem \ref{thm:1}, we need to prove that: 
there exists a choice of $m_{\text{R,-}}$ for which the Euclidean correlations of 
the reference model with Hamiltonian $\mathcal H^{\text R}$, see \eqref{cc}, coincide with those of the original Hamiltonian $\mathcal H$; the Euclidean Kubo conductivity 
coincides with the real-time one. Cf. with the last two items, (iii) and (iv), of the list after \eqref{eq.mR-}. We also need to prove the regularity and symmetry properties of the critical curves, stated in Theorem \ref{thm:1}.

\medskip

Let us start with discussing item (iii), as well as the $C^1$ regularity of the critical curves. 
In order to prove the equivalence of $\mathcal H$ and $\mathcal H^{\text R}$, it is enough to fix the counterterms as discussed in the previous section, and choose $m_{\text{R},-}$ to be the solution of 
\eqref{eq.mR-}. 
%
Let us then show that \eqref{eq.mR-} can be inverted in the form $m_{R,-}= m_{R,-}(U, W, \phi)$, with $m_{R,-}(U, W, \phi)$ analytic in $U$ and $C^{1}$ in $W, \phi$. 
We want to appeal to the analytic implicit function theorem. For this purpose, we need to estimate the derivative of $\d(U,m_{\text{R,-}},\phi)$ w.r.t. $m_{\text{R},-}$. 
Recall that $\d_{-,0}=\d_{-,0}(U,m_{\text{R},-},\phi)$ satisfies the second of \eqref{in.dat}, and that $\d(U,m_{\text{R},-},\phi)$ and $\d_{-,0}(U,m_{\text{R},-},\phi)$ are analytically close (they differ only 
because of the effect of the ultraviolet integration). Therefore, 
$$ \d(U,m_{R,-},\phi)=-\sum_{k=h^*_2+1}^1 2^{k-1}\b^\d_{-,k}(U,\underline\t),$$
where $\b^\d_{-,k}(U,\underline\t)$ accounts for the difference between $\d$ and $\d_{-,0}$ due to the ultraviolet integration. Differentiating both sides with respect to the mass, we find: 
$$\frac{\partial \d(U,m_{R,-},\phi)}{\partial m_{R,-}}=-\sum_{k=h^*_2+1}^1 2^{k-1}\frac{\partial \b^\d_{-,k}}{\partial m_{R,-}}(U,\underline\t),$$
which should be looked at as (a component of) a fixed point equation for the derivatives of the running coupling constants, analogous to the ones solved in the proof of Lemma \eqref{lem:count}. 
When acting on the beta function, the derivative with respect to $m_{R,-}$ can act on a propagator $g^{(k')}_\o$, or on a running coupling constant. When acting on a propagator, it replaces $g^{(k')}_\o$
by $\frac{\partial g^{(k')}_\o}{\partial m_{\text{R},-}}$, which is bounded dimensionally in the same way as $g^{(k')}_\o$, times an extra factor proportional to $2^{-k'}$. On the other hand, the action of the derivative on a 
running coupling constant should be bounded inductively, in the same spirit as the proof of Lemma \ref{lem:count}. All in all, recalling also the basic bound on the beta function, \eqref{eq.bound}, we get
\be\label{eq:derm}
\Big|\frac{\partial\d(U,m_{R,-},\phi)}{\partial m_{R,-}}\Big|\le \sum_{k=h^*_2+1}^1 2^{k}C_0|U|2^{\th k}2^{-k}\le C_2|U|,
\ee
for a suitable constant $C_2$. 
Exactly the same argument and estimates are valid for the derivative with respect to $\phi$, so that 
\be\label{eq:derph}
\Big|\frac{\partial \d(U,m_{R,-},\phi)}{\partial \phi}\Big|\le C_2|U|\;.
\ee
The last estimate is optimal for small $\phi$. For larger values of $\phi$, one can also take advantage of the symmetry under exchange $\phi\to\pi-\phi$ (the `magnetic reflections', see \eqref{eq.4.12}) to conclude 
that the derivative of $\d$ with respect to $\phi$ vanishes continuously as $\phi\to(\pi/2)^-$. Moreover, by the symmetry properties of the model, $\d(U,0,0)=0$. Therefore, $|\d(U, m_{\text{R}, -}, \phi)|\leq 2C_2|U| (|m_{\text{R},-}|+
\sin\phi)$. 

Using these bounds and the implicit function theorem, we see that \eqref{eq.mR-} can be inverted in the form \eqref{eq.2.23}, with $|\widetilde\d(U,W,\phi)|\leq C|U| (W+\sin\phi)$ for some constant $C$. 
The equation for the critical curve in the parameter range we are considering is simply $m_{\text{R,-}}=0$, that is $W=3\sqrt3t_2\sin\phi +\d(U,0,\phi)$, which is $C^1$ in $\phi$ and, thanks to the symmetries of the problem,  
it satisfies the properties stated in 
Theorem \ref{thm:1}.

\bigskip

We are left with discussing item (iv), that is, the equivalence between the Euclidean and real-time Kubo conductivities. Given our bounds on the Euclidean correlations, the equivalence follows from result discussed in previous papers. In fact, our bounds imply that the current-current correlations, at large space-time separations, decay either faster-than-any-power decay, if $m_{\text{R},-}\neq0$, or like $|\xx-\yy|^{-4}$, otherwise: therefore, we can repeat step by step the proof of \cite[Theorem 3.1]{GMPhall}, as the reader can easily check. For a slightly modified and simplified proof, see also \cite[Appendix B]{AMP} and \cite[Section 5]{MP2}. 

\medskip

This concludes the proof of Theorem \ref{thm:1}. \qed

\subsection{Concluding remarks}

In conclusion, the universality of the Hall conductivity (i.e., its independence from the interaction strength) can be seen as a consequence of lattice conservation laws, combined with the regularity properties of the correlation 
functions. The quantization of the interacting Hall conductivity then follows from its quantization in the non-interacting case: however, an important point in the proof is to compare the interacting system and its conductivity with the right reference non-interacting system, that is, the one with the right value of the mass; this is the reason why we introduce a reference non-interacting system with mass equal to the renormalized mass of the interacting system;
in order to fix the correct value of the renormalized mass, we need to solve a fixed point equation for it. 
The same strategy we proposed in the present context can be easily extended to prove that the Hall conductivity is constant against {\em any} deformation of the Hamiltonian, provided the Fourier transform of the correlations are sufficiently regular at each point of the interpolation, in the sense specified by Proposition \ref{assm}. Note that our universality result is valid as soon as the Fourier transform of the current-current-interaction correlations are 
$C^3$ in momentum space (a critical analysis of the proof shows that we need even less: $C^{2+\e}$ with $\e>0$ is a sufficient condition for our construction to work). This means that we do not require the existence of a spectral gap, or the exponential decay of correlations: sufficiently fast polynomial decay is actually enough. It would be nice to provide a realistic example of a gapless model exhibiting a non-trivial, universal behavior of the transverse conductivity (or, in alternative, to exclude the possibility that such a model exists). 

\medskip

A problem connected with the one discussed in this paper, but much more challenging, is to prove universality of the conductivity for massless models with slow polynomial decay of correlations: by `slow', here, we mean that 
Proposition \ref{assm} cannot be applied. A first example is the Haldane model, considered in this paper, for values of the parameters {\it on} the critical line. In this case, as already recalled after the statement of Theorem 
\eqref{thm:1}, one can prove the universality of the {\it longitudinal conductivity} \cite{GJMP}: the proof, which generalizes the one in \cite{GMP1}, uses lattice Ward Identities, combined with the symmetry properties of the 
current-current correlation functions. It would be very interesting to establish the universality, or the violation thereof, of the transverse conductivity on the critical line. 

Another gapless context, where the issue of the universality of the conductivity naturally arises, is the case of bulk massive systems in non-trivial domains with, say, Dirichlet conditions imposed at the boundary. 
In such a setting, usually, massless edge states appear, and the edge system is characterized by correlations with slow polynomial decay. 
Nevertheless, universality holds as a consequence of more subtle mechanisms, which relies on the non-renormalization of the edge chiral anomaly. Using these ideas, two of us proved the validity of the bulk-edge correspondence 
in lattice Hall systems with single-mode chiral edge currents \cite{AMP}, and in the spin-conserving Kane-Mele model \cite{MP}. It would be very interesting to generalize these findings to lattice systems with several 
edge modes, as well as to continuum systems. 

%
%

\bigskip

\noindent{{\bf Acknowledgements.}} The work of A. G. and of V. M. has been supported by the European Research Council (ERC) under the European Union's Horizon 2020 research and innovation programme (ERC CoG UniCoSM, grant agreement n.724939). V. M. acknowledges support also from the Gruppo Nazionale di Fisica Matematica (GNFM). The work of M. P. has been supported by the Swiss National Science Foundation, via the grant ``Mathematical Aspects of Many-Body Quantum Systems''.

\end{document}

%% file: honeycomb5b.pdf_tex
\begingroup%
  \makeatletter%
  \providecommand\color[2][]{%
    \errmessage{(Inkscape) Color is used for the text in Inkscape, but the package 'color.sty' is not loaded}%
    \renewcommand\color[2][]{}%
  }%
  \providecommand\transparent[1]{%
    \errmessage{(Inkscape) Transparency is used (non-zero) for the text in Inkscape, but the package 'transparent.sty' is not loaded}%
    \renewcommand\transparent[1]{}%
  }%
  \providecommand\rotatebox[2]{#2}%
  \ifx\svgwidth\undefined%
    \setlength{\unitlength}{949.63317086bp}%
    \ifx\svgscale\undefined%
      \relax%
    \else%
      \setlength{\unitlength}{\unitlength * \real{\svgscale}}%
    \fi%
  \else%
    \setlength{\unitlength}{\svgwidth}%
  \fi%
  \global\let\svgwidth\undefined%
  \global\let\svgscale\undefined%
  \makeatother%
  \begin{picture}(1,1.04878785)%
    \put(0,0){\includegraphics[width=\unitlength,page=1]{honeycomb5b.pdf}}%
    \put(0.86226951,0.90481402){\color[rgb]{0,0,0}\makebox(0,0)[lb]{\smash{$B$}}}%
    \put(0.69599946,0.90481402){\color[rgb]{0,0,0}\makebox(0,0)[lb]{\smash{$A$}}}%
    \put(0,0){\includegraphics[width=\unitlength,page=2]{honeycomb5b.pdf}}%
    \put(0.42233072,0.29499993){\color[rgb]{0,0,0}\makebox(0,0)[lb]{\smash{$\vec x$}}}%
    \put(0,0){\includegraphics[width=\unitlength,page=3]{honeycomb5b.pdf}}%
    \put(0.4221171,0.12414161){\color[rgb]{0,0,0}\makebox(0,0)[lb]{\smash{$\vec \g_{1}$}}}%
    \put(0.50653636,0.48973496){\color[rgb]{0,0,0}\makebox(0,0)[lb]{\smash{$\vec\g_{2}$}}}%
    \put(0.24143859,0.35129122){\color[rgb]{0,0,0}\makebox(0,0)[lb]{\smash{$\vec\g_{3}$}}}%
    \put(0,0){\includegraphics[width=\unitlength,page=4]{honeycomb5b.pdf}}%
    \put(0.74744947,0.61477123){\color[rgb]{0,0,0}\makebox(0,0)[lb]{\smash{$\vec\ell_{2}$}}}%
    \put(0.7450444,0.39073242){\color[rgb]{0,0,0}\makebox(0,0)[lb]{\smash{$\vec\ell_{1}$}}}%
  \end{picture}%
\endgroup%